\documentclass[12pt]{article}

\usepackage{adjustbox}
\usepackage[left=2cm, right=2cm, top=2cm, bottom=2cm]{geometry}
\usepackage{enumitem}
\usepackage{graphicx}
\usepackage{subfigure}
\usepackage{bm}
\usepackage{amsmath, amsfonts, amssymb, amsthm, mathtools}
\usepackage{tikz}
\usepackage{array}
\usepackage{multirow}
\usepackage{cancel}
\usepackage{appendix}
\usetikzlibrary{decorations.pathreplacing}
\newtheorem{theorem}{Theorem}[section]
\newtheorem{lemma}{Lemma}[section]

\newcommand\Tau{\mathcal{T}}

\usepackage{amssymb}

\usepackage[figuresright]{rotating}

\begin{document}

\title{A Weight Function Lemma Heuristic\\ for Graph Pebbling \footnote{This work was supported by CAPES, CNPq, and FAPERJ.}}




\author{G. A. Bridi, F. L. Marquezino, C. M. H. de Figueiredo\\
Federal University of Rio de Janeiro, Brazil\\
\{gabridi, franklin, celina\}@cos.ufrj.br
}

\date{}

\maketitle

\begin{abstract}

Graph pebbling is a problem in which pebbles are distributed across the vertices of a graph and moved according to a specific rule: two pebbles are removed from a vertex to place one on an adjacent vertex. The goal is to determine the minimum number of pebbles required to ensure that any target vertex can be reached, known as the pebbling number. Computing the pebbling number lies beyond NP in the polynomial hierarchy, leading to bounding methods. One of the most prominent techniques for upper bounds is the Weight Function Lemma (WFL), which relies on costly integer linear optimization. To mitigate this cost, an alternative approach is to consider the dual formulation of the problem, which allows solutions to be constructed by hand through the selection of strategies given by subtrees with associated weight functions. To improve the bounds, the weights should be distributed as uniformly as possible among the vertices, balancing their individual contribution. However, despite its simplicity, this approach lacks a formal framework. To fill this gap, we introduce a novel heuristic method that refines the selection of balanced strategies. The method is motivated by our theoretical analysis of the limitations of the dual approach, in which we prove lower bounds on the best bounds achievable. Our theoretical analysis shows that the bottleneck lies in the farthest vertices from the target, forcing surplus weight onto the closer neighborhoods. To minimize surplus weight beyond the theoretical minimum, our proposed heuristic prioritizes weight assignment to the farthest vertices, building the subtrees starting from the shortest paths to them and then filling in the weights for the remaining vertices. Applying our heuristic to Flower snarks and Blanuša snarks, we improve the best-known upper bounds, demonstrating the effectiveness of a structured strategy selection when using the WFL.
\end{abstract}

\section{Introduction}

Graph pebbling is a mathematical game with origins in number theory that has evolved into a broader combinatorial study with applications not only in number theory, but also in optimization, graph theory, probability theory, and computer science~\cite{recent_survey}. The game can be seen as a transportation problem, where pebbles representing resources to be delivered to some destination are placed at the vertices of a graph and moved through the edges. Each pebbling move requires spending two pebbles, but only one reaches the adjacent vertex, like a toll paid along the way. The pebbling number of the graph is the minimum number of pebbles needed to guarantee that, no matter how the pebbles are placed on the vertices of the graph, at least one pebble can reach any chosen target vertex---the destination---by some combination of pebbling moves. 

The problem of graph pebbling is known to be $\Pi_2^P$-complete~\cite{pebbling_completeness2}, placing it beyond NP in the polynomial hierarchy---thus making the exact computation of the pebbling number infeasible for large graphs. Consequently, various techniques have been developed to delimit the pebbling number through bounds. One of them is the Weight Function Lemma (WFL), introduced by Hurlbert~\cite{WFL}, which provides a useful method for deriving upper bounds. As discussed in Hurlbert~\cite{WFL}, applying WFL is not a trivial task, since it involves an integer linear optimization problem whose optimal solution provides a direct upper bound for the pebbling number. Although the associated optimization problem can be solved for small graphs---and since in some cases the relaxation is sufficient, somewhat larger graphs can be handled---the problem becomes very large when scaling the size of the graphs. In general, solving integer linear problems is NP-complete~\cite{ILP}.

An alternative approach, explored in Refs.~\cite{WFL_snarks, WFL_kneser, pebbling_split}, is to work with the dual problem to obtain (at least) near-optimal upper bounds, as it can yield solutions with short certificates in some cases. A key advantage of the dual formulation is the possibility of constructing solutions without relying on optimization techniques, allowing certificates to be obtained by hand, leveraging the combinatorial structure of the graph. Specifically, obtaining a solution to the dual problem involves selecting a set of strategies, where each strategy consists of a subtree rooted at the target vertex and is associated with a weight function that assigns non-negative values to each vertex in the graph. Better bounds are obtained when the sum of the weights of each of the strategies is distributed as evenly as possible among the vertices, balancing the overall contribution required from each vertex. The strength of this combinatorial approach lies in its simplicity. However, it lacks a formal framework to guide its systematic application.

In this sense, we introduce a novel heuristic method to build strategies, motivated by a theoretical analysis of the limitations of the dual formulation of the WFL. Our analysis establishes lower bounds on the best achievable upper bounds in this formulation, revealing structural sources of imbalance in the weight distribution over the vertices. These insights guide the design of our heuristic, which we apply to the Blanuša and Flower snarks, improving the upper bounds recently obtained by Adauto et al.~\cite{WFL_snarks}. Their work is a pioneering study on the pebbling number of snark graphs, and by focusing on the same graph families, we provide a direct and fair comparison between our heuristic results and the state of the art. Our results highlight the potential of guided strategy construction for improving bounds in challenging graph families.

\paragraph{Snarks}

In this work, we assume that $G = (V, E)$ is a simple connected graph. The number of vertices of $G$ is $n(G)$. The distance between vertices $u$ and $v$ in a graph $G$ is denoted by $d_G(u, v)$. The eccentricity of a vertex $v$, represented by $e(v)$, is the maximum distance $d_G(u, v)$ for any $u \in V(G)$. The largest $e(v)$ of $G$ is the diameter $D(G)$. For $j \in \mathbb{N}$ and a vertex $v$, the $j^{\text{th}}$ neighborhood of $v$, denoted by $N_j(v)$, is the set of vertices $u$ such that $d_G(u, v) = j$. In particular, we denote the first neighborhood (or simply the neighborhood) of $v$ as $N(v)$, while the $e(v)^{\text{th}}$ neighborhood of $v$, expressed as $P(v)$, is the set of $v$-peripheral vertices.

Snarks~\cite{survey_snarks} are a family of graphs that are cubic, bridgeless, $4$-edge-chromatic, defined in the context of the Four Color Theorem, which holds if and only if no snark is planar~\cite{fourColors}. The Petersen graph, discovered in 1898~\cite{petersen}, is the first and smallest snark, with $10$ vertices. Shortly after, in 1946, Danilo Blanuša introduced the Blanuša 1 and 2 snarks, called $B_1$ and $B_2$~\cite{survey_snarks}, both with $18$ vertices and diameter~$4$. The present work considers the $B_2$ snark, shown in Figure~\ref{fig:B2_x3}(a). There are no snarks of order $12$, $14$, $16$, whereas snarks exist for any even order greater than $16$. For a detailed history, see Ref.~\cite{survey_snarks}. We also consider the family of Flower snarks $J_m$~\cite{Campos}, defined for odd $m = 2k+1, m \geq 3$. For each $i\in\pm\{0,1,\ldots,k\}$, we have vertices $v_i$, $x_i$, $y_i$, and $z_i$. The vertices $v_i$, $x_i$, and $y_i$ are all adjacent to $z_i$, while the vertices $\{v_i\}$ form the cycle $C_m$, with adjacencies given by consecutive indices modulo $m$. The vertices $\{x_i\}$ (resp. $\{y_i\}$) form a path given by the cycle with the edge $x_kx_{-k}$ (resp. $y_ky_{-k}$) removed. Finally, we add the edges $x_ky_{-k}$ and $y_kx_{-k}$ so that the two paths now form one cycle $C_{2m}$. We have $n(J_m) = 4m$ and $D(J_m) = k+2$.

\paragraph{Graph pebbling}
A configuration $C$ on a graph $G$ is a function $C: V(G) \rightarrow \mathbb{N}$ such that $C(v)$ is the number of pebbles at the vertex $v \in V(G)$. The total number of pebbles in the graph, $\sum_{v \in V(G)} C(v)$, is denoted by $|C|$. A pebbling move consists of removing $2$ pebbles from a vertex $v$ and adding $1$ pebble to a vertex $u \in N(v)$. From the configuration $C$, the goal of the graph pebbling problem is to get a combination of pebbling moves that reaches a target vertex $r \in V(G)$, i.e., that places a pebble on $r$. If such a combination of moves exists, we say that the configuration $C$ is $r$-solvable. Otherwise, $C$ is $r$-unsolvable. The problem of deciding on the solvability of a configuration is known to be NP-Complete~\cite{pebbling_completeness1, pebbling_completeness2}. The pebbling number of the target vertex $r$, denoted by $\pi(G, r)$, is the minimum number of pebbles $t$ such that all configurations with $t$ pebbles are $r$-solvable. The pebbling number of the graph, represented by $\pi(G)$, is the minimum number of pebbles $t$ such that all configurations with $t$ pebbles are $r$-solvable for all $r \in V(G)$, i.e., $\pi(G) = \max_{r \in V(G)} \{ \pi(G,r) \}$. For any graph, the pebbling number is delimited by the well-known basic bounds $\max \{ n(G), 2^{D(G)} \} \leq \pi(G) \leq (n(G) - D(G))(2^{D(G)} - 1) + 1$~\cite{survey_pebbling, basic_bound}. 

The WFL is stated in Lemma~\ref{l:WFL}. Before we proceed, we need some definitions. Let $T$ be a subtree of $G$ rooted in the target vertex $r \in V(G)$ such that $n(T) \geq 2$. The subtree $T$ is called an $r$-strategy when it is associated with a non-negative function $\omega_T$, called weight function, such that $\omega_T(v) = 0$ if $v \notin V(T)$ or $v = r$, and $\omega_T(u) \geq 2 \omega_T(v)$ for any pair of vertices $u, v \in V(T) - \{r\}$ such that $u$ is the parent of $v$ in $T$. We denote the sum of the weights on the vertices of the subtree as $|\omega_T| = \sum_{v \in V(G) - \{r \}} \omega_T(v)$. The definition of the weight extends to configurations as, for a configuration $C$, $\omega_T(C) = \sum_{v \in V(G) - \{r \}} \omega_T(v) C(v)$. 

\begin{lemma} [Weight Function Lemma~\cite{WFL}] \label{l:WFL}
Let $T$ be a $r$-strategy of the graph $G$ with associated weight function $\omega_T$ and let $C$ be an $r$-unsolvable configuration on $G$. Then, $\omega_T(C) \leq |\omega_T|$.
\end{lemma}

The upper bound on the pebbling number achievable by WFL is given by $\pi(G, r) \leq \zeta_{G, r} + 1$, where $\zeta_{G, r}$ is the optimal solution of the integer linear optimization problem, given by the maximization of $\sum_{v \in V(G) - \{ r\}} C(v)$ subject to $\omega_T(C) \leq |\omega_T| \ \forall T \in \Tau_{\text{set}}$, with $\Tau_{\textit{set}}$ being the set of all possible $r$-strategies of $G$. Note that only the optimal solution can generate a valid upper bound on the pebbling number. In contrast, we can get upper bounds from any solution of the associated dual problem, given by the minimization of $\sum_{T \in \Tau_{\text{set}}} \alpha_T |\omega_T|$ subject to $\sum_{T \in \Tau_{\text{set}}} \alpha_T \omega_T(v) \geq 1 \ \forall v \in V(G) - \{r \}$, where $\alpha_T \geq 0$ for all $T \in \Tau_{\text{set}}$. 

We can manipulate the dual problem as follows. Let $\Tau \subseteq \Tau_{\text{set}}$ be the set of subtrees in which $\alpha_T > 0$. Since we can rescale the weight function by a constant factor, let $1/\alpha = \alpha_T$ for all $T \in \Tau$. Extending the weight function to the set $\Tau$, we have $\omega_\Tau(v) = \sum_{T \in \Tau} \omega_T(v)$ and $|\omega_\Tau| = \sum_{T \in \Tau} |\omega_T|$. The dual problem is reduced to minimizing $|\omega_\Tau|/\alpha$ subject to $\omega_{\Tau}(v) \geq \alpha \ \forall v \in V(G) - \{r \}$. The optimal value of $\alpha$ is $\omega_{\text{min}} = {\min}_{v \in V(G) - \{ r \}} \{ \omega_\Tau(v) \}$---the minimum weight---and we need to minimize $\lambda_\Tau = |\omega_\Tau|/\omega_{\text{min}}$. The quantity $\lambda_\Tau$, which we call the WFL ratio of $\Tau$, provides the upper bound $\pi(G, r) \leq \lfloor \lambda_\Tau \rfloor + 1$. The optimal solution of the dual problem $\lambda_{G, r} = {\min}_{\Tau \subseteq \Tau_{\textit{set}}} \{ \lambda_\Tau \}$---the minimum WFL ratio---provides the best dual bound on the pebbling number of the target $r$, while $\lambda_{G} = {\max}_{r \in V(G)} \{ \lambda_{G, r} \}$ does the same for the graph $G$. The lower bounds on the pebbling number result in direct lower bounds on the WFL ratio. For instance, $\pi(G) \geq \max \{ n(G), 2^{D(G)} \}$ implies that $\lambda_G \geq \max \{ n(G), 2^{D(G)} \} - 1$. Figure~\ref{fig:petersen} illustrates the dual approach applied to the Petersen graph.

\begin{figure}[ht]
    \centering
    \begin{minipage}{0.32\textwidth}
        \centering
        \subfigure[]{
            \begin{tikzpicture}[scale=0.78]
            
            \tikzstyle{every node}=[circle, draw, fill=black, inner sep=1.5pt]
            
            \node (A1) at (-0.00, 2.00) [label=above:{\small $r$}] {}; 
            \node (A2) at (-1.90, 0.62) [label=left:{\small $4$}]{};
            \node (A3) at (-1.18, -1.62) [label=left:{\small $2$}] {};
            \node (A4) at (1.18, -1.62) [label=right:{\small $1$}]{};
            \node (A5) at (1.90, 0.62) {};
            \node (B1) at (-0.00, 1.00) {};
            \node (B2) at (-0.95, 0.31) [label=above:{\small $2$}] {};
            \node (B3) at (-0.59, -0.81) [label=left:{\small $1$}] {};
            \node (B4) at (0.59, -0.81) [label=right:{\small $1$}]{};
            \node (B5) at (0.95, 0.31) [label=above:{\small $1$}]{};
    
            \draw (A1) -- (A2);
            \draw (A2) -- (A3);
            \draw (A3) -- (A4);
            \draw[dashed] (A4) -- (A5);
            \draw[dashed] (A5) -- (A1);
            \draw[dashed] (B1) -- (B3);
            \draw[dashed] (B3) -- (B5);
            \draw (B5) -- (B2);
            \draw (B2) -- (B4);
            \draw[dashed] (B4) -- (B1);
            \draw[dashed] (A1) -- (B1);
            \draw (A2) -- (B2);
            \draw (A3) -- (B3);
            \draw[dashed] (A4) -- (B4);
            \draw[dashed] (A5) -- (B5);
            
            \end{tikzpicture}
        }
    \end{minipage}
    \begin{minipage}{0.32\textwidth}
        \centering
        \subfigure[]{
            \begin{tikzpicture}[scale=0.78]

            \tikzstyle{every node}=[circle, draw, fill=black, inner sep=1.5pt]
            
            \node (C1) at (-6.00, 2.00) [label=above:{\small $r$}]{};
            \node (C2) at (-7.90, 0.62) {}; 
            \node (C3) at (-7.18, -1.62) [label=left:{\small $1$}]{};
            \node (C4) at (-4.82, -1.62) [label=right:{\small $1$}]{};
            \node (C5) at (-4.10, 0.62) {};
            \node (D1) at (-6.00, 1.00) [label=left:{\small $4$}] {};
            \node (D2) at (-6.95, 0.31) [label=above:{\small $1$}]{};
            \node (D3) at (-6.59, -0.81) [label=left:{\small $2$}]{};
            \node (D4) at (-5.41, -0.81) [label=right:{\small $2$}]{};
            \node (D5) at (-5.05, 0.31) [label=above:{\small $1$}]{};
    
            \draw[dashed] (C1) -- (C2);
            \draw[dashed] (C2) -- (C3);
            \draw[dashed] (C3) -- (C4);
            \draw[dashed] (C4) -- (C5);
            \draw[dashed] (C5) -- (C1);
            \draw (D1) -- (D3);
            \draw (D3) -- (D5);
            \draw[dashed] (D5) -- (D2);
            \draw (D2) -- (D4);
            \draw (D4) -- (D1);
            \draw (C1) -- (D1);
            \draw[dashed] (C2) -- (D2);
            \draw (C3) -- (D3);
            \draw (C4) -- (D4);
            \draw[dashed] (C5) -- (D5);
            
            \end{tikzpicture}
        }
    \end{minipage}
    \begin{minipage}{0.32\textwidth}
        \centering
        \subfigure[]{
            \begin{tikzpicture}[scale=0.78]
            
            \tikzstyle{every node}=[circle, draw, fill=black, inner sep=1.5pt]

            \node (E1) at (-6.00, 2.00) [label=above:{\small $r$}]{};
            \node (E2) at (-7.90, 0.62) {};
            \node (E3) at (-7.18, -1.62) [label=left:{\small $1$}]{};
            \node (E4) at (-4.82, -1.62) [label=right:{\small $2$}]{};
            \node (E5) at (-4.10, 0.62) [label=right:{\small $4$}]{};
            \node (F1) at (-6.00, 1.00) {};
            \node (F2) at (-6.95, 0.31) [label=above:{\small $1$}]{};
            \node (F3) at (-6.59, -0.81) [label=left:{\small $1$}]{};
            \node (F4) at (-5.41, -0.81) [label=right:{\small $1$}]{};
            \node (F5) at (-5.05, 0.31) [label=above:{\small $2$}]{};
    
            \draw[dashed] (E1) -- (E2);
            \draw[dashed] (E2) -- (E3);
            \draw (E3) -- (E4);
            \draw (E4) -- (E5);
            \draw (E5) -- (E1);
            \draw[dashed] (F1) -- (F3);
            \draw (F3) -- (F5);
            \draw (F5) -- (F2);
            \draw[dashed] (F2) -- (F4);
            \draw[dashed] (F4) -- (F1);
            \draw[dashed] (E1) -- (F1);
            \draw[dashed] (E2) -- (F2);
            \draw[dashed] (E3) -- (F3);
            \draw (E4) -- (F4);
            \draw (E5) -- (F5);
            
            \end{tikzpicture}
        }
    \end{minipage}
    \caption{Adapted representation from Ref.~\cite{WFL} of the optimal $r$-strategies highlighted in solid lines (a) $T_1$, (b) $T_2$, and (c) $T_3$ for the Petersen graph $P$. The Petersen graph is vertex-transitive, allowing us to consider the vertex $r$ as the only considered target. Note that $|\omega_\Tau| = 36$, $\omega_{\text{min}} = 4$, and $\lambda_\Tau = 9$. Consequently, $\pi(P) \leq 10$. Since $n(P) = 10$, it follows that $\pi(P) = 10$ and $\lambda_P = 9$.}
    \label{fig:petersen}
\end{figure}
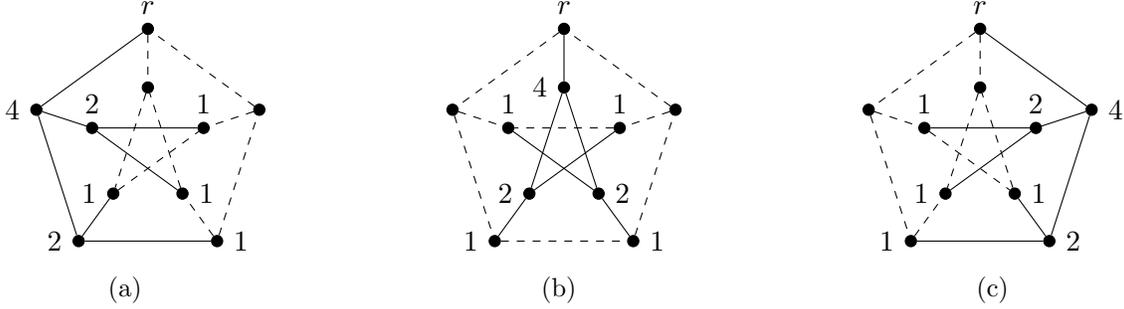

\section{Limits of Weight Function Lemma}

Minimizing the WFL ratio $\lambda_\Tau$ involves finding a set of strategies $\Tau$, that we call certificate, that better balances the weight of the function $\omega_\Tau$ over the vertices, 
making $\omega_\Tau$ as uniform as possible. Assuming the most favorable scenario, i.e., $\omega_\Tau(v) = \omega_{\text{min}} \ \forall v \in V(G) - \{ r \}$, gives the trivial lower bound $\lambda_{G, r} \geq n(G) - 1$. When there is a certificate $\Tau$ with this property---such as the case of the Petersen graph---we have that $\pi(G, r) = n(G)$, a fact known as the Uniform Covering Lemma~\cite{WFL}. Observe that we can write the WFL ratio $\lambda_\Tau$ by adding to $n(G) - 1$ the weight at each vertex that exceeds $\omega_{\text{min}}$ and then divide by $\omega_{\text{min}}$. That way,
\begin{equation}
    \label{eq:sur}
\lambda_\Tau = n(G) - 1 + \frac{\sum_{v \in V(G) - \{r \}} (\omega_\Tau(v) - \omega_{\text{min}})}{\omega_{\text{min}}}.
\end{equation}
We call $(\omega_\Tau(v) - \omega_{\text{min}})/\omega_{\text{min}}$  the surplus weight of the vertex $v$. To obtain a non-trivial lower bound on the minimum WFL ratio $\lambda_{G, r}$, given by Theorem~\ref{thm:min}, we use the relationship between the weight of $r$-peripheral vertices and the vertices of the closer neighborhoods. For each neighborhood $N_j(r)$, $1 \leq j \leq e(r) - 1$, in which $2^{e(r) - j} > |N_j(r)|$, we get a non-zero lower bound on the surplus weight. We denote by $I_{\text{sur}}$ the set of such indices $j$, i.e.,  $I_{\text{sur}} = \{1 \leq j \leq e(r) - 1: 2^{e(r) - j} > |N_j(r)| \}$.
 
\begin{theorem} \label{thm:min}
The minimum WFL ratio of a target vertex $r$ on a graph $G$ is bounded by
\begin{equation}
    \label{eq:low}
\lambda_{G, r} \geq n(G) - 1 + \sum_{j \in I_{\text{sur}}} (2^{e(r) - j} - |N_j(r)|).
\end{equation}
\end{theorem}

\begin{proof}
Since $\omega_\Tau(v) \geq \omega_{\text{min}} \ \forall v \in V(G) - \{ r \}$, we can bound $\lambda_{G, r}$ from Eq.~\eqref{eq:sur} by
\begin{equation}
    \label{eq:proof}
\lambda_{G, r} \geq n(G) - 1 + \frac{\sum_{j =  1}^{e(r) - 1} \sum_{v \in N_j(r)} (\omega_\Tau(v) - \omega_{\text{min}})}{\omega_{\text{min}}} = n(G) - 1 + \sum_{j =  1}^{e(r) - 1} \left( \frac{\sum_{v \in N_j(r)} \omega_\Tau(v)}{\omega_{\text{min}}} - |N_j(r)| \right).
\end{equation}
Let $u \in P(r)$ and $1 \leq j \leq e(r) - 1$. Note that for each subtree $T \in \Tau$ in which $\omega_T(u) > 0$, there is a vertex $v_j \in N_j(r)$ in which $\omega_T(v_j) \geq 2^{e(r) - j} \omega_T(u)$. Thus, $\sum_{v \in N_j(r)} \omega_\Tau(v) \geq 2^{e(r) - j} \omega_{\text{min}}$. On the other hand, $\sum_{v \in N_j(r)} \omega_\Tau(v) \geq |N_j(r)| \omega_{\text{min}}$. Combining both lower bounds on Eq.~\eqref{eq:proof}, the theorem follows.
\end{proof}

For each $j \in I_{\text{sur}}$, Theorem~\ref{thm:min} counts the minimum surplus weight in the $j^{\text{th}}$ neighborhood that originates from the path between the vertices of this neighborhood and the $r$-peripheral vertices in the subtrees. From Eq.~\eqref{eq:low}, we see that the minimum surplus weight on $N_j(r)$ grows exponentially with the distance between the vertices of $N_j(r)$ and the $r$-peripheral vertices, a dependency that matches the exponential nature of the weight function along paths. Furthermore, this growth can be linearly counterbalanced in the number of vertices of $N_j(r)$, an intuitive result, since more neighbors imply more paths to distribute and compensate for the weight imbalance. Specifically, if $|N_j(r)| \geq 2^{e(r) - j}$, surplus weight is not strictly required in the $j^{\text{th}}$ neighborhood. If it holds for every $j$, i.e., $I_{\text{sur}} =\emptyset$, surplus weight is not strictly required in the whole graph. This is the case of the Petersen graph, with the strategies shown in Figure~\ref{fig:petersen}.  

To achieve the minimum surplus weight of Theorem~\ref{thm:min} and thus get the tight bound of Eq.~\eqref{eq:low}, we have the following necessary condition. For any $j \in I_{\text{sur}}$, $v_j \in N_j(r)$, and $u \in P(r)$, there must be a shortest path from $r$ to $u$ passing by $v_j$. Equivalently, $d_{G - \{r\}}(v_j, u) = e(r) - j$ must hold for any $j, v_j, u$. To indeed build the tight bound certificate, these shortest paths must be incorporated into the subtrees so that for each vertex $v_j$, there must be at least one subtree in which $v_j$ is part of at least one shortest path from $r$ to all $u \in P(r)$. In addition, we need to ensure that no more surplus weights are added in the subtrees beyond those provided by Theorem~\ref{thm:min}. In particular, the weight of every vertex $v_j \in N_j(r)$, $j \notin I_{\text{sur}}$, need to be $\omega_{\text{min}}$, while for $j \in I_{\text{sur}}$, the sum of the weight in $N_j(r)$ must be $2^{e(r) - j} \omega_{\text{min}}$. An example where we can achieve the tight bound of Eq.~\eqref{eq:low} is the target vertex $z_0$ on Flower graph $J_3$, in which the certificate is shown in Figure~\ref{fig:Jm_z0}(a). Specifically, $J_3$ is a graph in which $I_{\text{sur}} =\{1\}$, a particular case that plays a determining role in the heuristic method of Section~\ref{sec:heu}. For a more general context, we have the $d$-dimensional cube graph $Q_d$, whose tight bound certificate can be found in the demonstration of Theorem $16$ of Hurlbert~\cite{WFL} paper.

There is an alternative way of expressing the lower bound of Theorem~\ref{thm:min}. Specifically, let $I_{\text{no}} =  \{1 \leq j \leq e(r): |N_j(r)| > 2^{e(r) - j} \}$, the set of indices $j$ in which we do not establish a non-zero minimum surplus weight for the $j^{\text{th}}$ neighborhood (excluding the equality case $2^{e(r) - j} = |N_j(r)|$ and including $P(r)$). Eq.~\eqref{eq:low} can be written in terms of $I_{\text{no}}$ as
\begin{equation}
    \label{eq:alt}
\lambda_{G, r} \geq 2^{e(r)} - 1 + \sum_{j \in I_{\text{no}}} (|N_j(r)| - 2^{e(r) - j}).
\end{equation}
Note that, combining Eqs.~\eqref{eq:low} and~\eqref{eq:alt}, we can conclude that our lower bound is stronger than the basic bound $\lambda_{G} \geq \max\{n(G), 2^{D(G)}\} - 1$. This alternative expression can be useful to simplify the calculations on graphs with targets in which $|I_{\text{sur}}| \gg |I_{\text{no}}|$, such as in graphs with a large diameter. In addition to its practical utility, Eq.~\eqref{eq:alt} also provides a complementary interpretation to Theorem~\ref{thm:min}. Specifically, $\lambda_{G, r} = 2^{e(r)} - 1$ is achievable when for every $1 \leq j \leq e(r)$, the inequality $ \sum_{v \in N_j(r)} \omega_{\mathcal{T}}(v) \geq 2^{e(r) - j} \omega_{\min} $ holds with equality. On the other hand, we also know that $\sum_{v \in N_j(r)} \omega_{\Tau}(v) \geq |N_j(r)| \omega_{\text{min}}$. Thus, if $|N(r)| \leq 2^{e(r) - j}$, in principle, there is room to distribute the total weight $2^{e(r) - j} \omega_{\text{min}}$ on the neighborhood vertices. Otherwise, i.e., if $|N(r)| > 2^{e(r) - j}$, the vertices of $N_j(r)$ force $\lambda_{G, r}$ to exceed $2^{e(r)} - 1$ proportional to the difference $|N_j(r)| - 2^{e(r) - j}$.

\section{The heuristic method} \label{sec:heu}

The technique used to prove the limitations of WFL on Theorem~\ref{thm:min} indicates that the bottleneck on the minimization of the WFL ratio $\lambda_\Tau$ is the $r$-peripheral vertices. These vertices, while potentially dictating the value of minimum weight $\omega_{\text{min}}$, force some vertices in the closer neighborhoods to have surplus weight. Consequently, to get close to the optimal certificate, we need to handle the assignment of weights to the $r$-peripheral vertices cautiously. This task involves striking a balance between two goals. The first is to assign as much weight as possible to the $r$-peripheral vertices, pushing the value of $\omega_{\text{min}}$ up. The second is to look for paths that reach the $r$-peripheral vertices while minimizing the surplus weight in the closer neighborhoods and, consequently, $|\omega_\Tau|$.

Using these insights, we developed a heuristic method to build suitable strategies for a useful application of WFL. Our method is based on the particular case $I_{\text{sur}} = \{1 \}$ of Theorem~\ref{thm:min}. i.e., based on the surplus weight of $N(r)$. Broadly speaking, we build one $r$-strategy for each neighbor $v$ of $r$. For each $v$, we start to build the subtree using shortest paths in the graph $G - \{r\}$ from $v$ to all the $r$-peripheral vertices, determining the minimum weight. Then, we complete our certificate $\Tau$ by adding more vertices to the subtrees to complete the minimum weight in all vertices. Formally, we define the following steps.
\begin{enumerate}
    \item Denote the neighbors of the target $r$ as $v_1, v_2, \ldots, v_{|N(r)|}$ and associate each $v_j$ with a $r$-strategy $T_j$ in such a way that the certificate is $\Tau = \{T_1, T_2, \ldots, T_{|N(r)|} \}$. For each $T_j \in \Tau$, $v_j$ is the only child of $r$.
    \item Set the minimum weight as 
    \begin{equation}
        \label{eq:min}
        \omega_{\text{min}} = \underset{u \in P(r)}{\min} \left\{ \sum_{j =  1}^{|N(r)|} 2^{e(r) - 1 - d_{G - \{r\}}(v_j, u)} \right\}.
    \end{equation}
    \item For all $j= 1, 2, \ldots, |N(r)|$, build the so-called trunk of the subtree $T_j$. The trunk of $T_j$ consists of adding shortest paths of the graph $G-\{r\}$ from $v_j \in N(r)$ to each $r$-peripheral vertex $u \in P(r)$. 
    At this step, for each given vertex $v$ that is added to the subtree $T_j$, assign the weight $\omega_{T_j}(v) = 2^{e(r) - 1 - d_{T_j}(v_j, v)}$.
    \item For each vertex $s$ whose summation of the weights on the $|N(r)|$ trunks built in step $3$ has not reached $\omega_{\text{min}}$, add $s$ to the subtrees until $\omega_{\Tau}(s) \geq \omega_{\text{min}}$. We call by branches the edges which we add to the subtrees in step $4$. If necessary, increase the weight of the trunk vertices.
\end{enumerate}

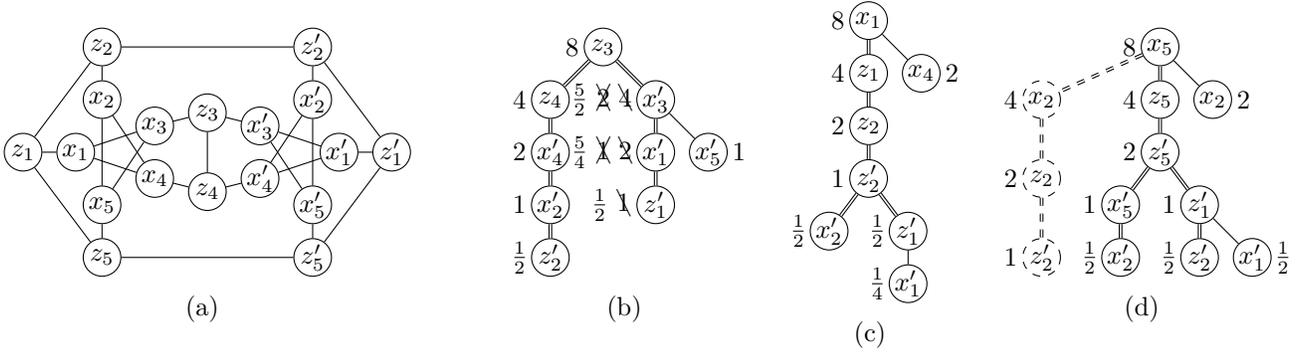
\begin{figure}[!t]
    \centering
    \hspace{-0.17\textwidth}
    \begin{minipage}{0.3\textwidth}
        \centering
        \subfigure[]{
            \begin{tikzpicture}[scale=0.35]
            \tikzstyle{every node}=[circle, draw, minimum size=0.5cm, inner sep=0pt]
            
            \node (A) at (0, 1.5) {\footnotesize $z_{3}$};
            \node (B) at (0, -1.5) {\footnotesize $z_{4}$};
            \node (C) at (-2, 1) {\footnotesize $x_{3}$};
            \node (D) at (-2, -1) {\footnotesize $x_{4}$};
            \node (E) at (2, 1) {\footnotesize $x_{3}'$};
            \node (F) at (2, -1) {\footnotesize $x_{4}'$};
            \node (G) at (-4, 2) {\footnotesize $x_{2}$};
            \node (H) at (-4, -2) {\footnotesize $x_{5}$};
            \node (I) at (-5, 0) {\footnotesize $x_{1}$};
            \node (J) at (4, 2) {\footnotesize $x_{2}'$};
            \node (K) at (4, -2) {\footnotesize $x_{5}'$};
            \node (L) at (5, 0) {\footnotesize $x_{1}'$};
            \node (M) at (-4, 4) {\footnotesize $z_{2}$};
            \node (N) at (-4, -4) {\footnotesize $z_{5}$};
            \node (O) at (-7, 0) {\footnotesize $z_{1}$};
            \node (P) at (4, 4) {\footnotesize $z_{2}'$};
            \node (Q) at (4, -4) {\footnotesize $z_{5}'$};
            \node (R) at (7, 0) {\footnotesize $z_{1}'$};

            \draw (A) -- (B);
            \draw (A) -- (C);
            \draw (B) -- (D);
            \draw (A) -- (E);
            \draw (B) -- (F);
            \draw (C) -- (H);
            \draw (C) -- (I);
            \draw (D) -- (G);
            \draw (D) -- (I);
            \draw (G) -- (H);
            \draw (E) -- (K);
            \draw (E) -- (L);
            \draw (F) -- (J);
            \draw (F) -- (L);
            \draw (J) -- (K);
            \draw (G) -- (M);
            \draw (H) -- (N);
            \draw (I) -- (O);
            \draw (J) -- (P);
            \draw (K) -- (Q);
            \draw (L) -- (R);
            \draw (M) -- (O);
            \draw (M) -- (P);
            \draw (N) -- (O);
            \draw (N) -- (Q);
            \draw (P) -- (R);
            \draw (Q) -- (R);
            \end{tikzpicture}
        }
    \end{minipage}
    \hspace{0.06\textwidth}
    \begin{minipage}{0.1\textwidth}
        \centering
        \subfigure[]{
            \begin{tikzpicture}[scale=0.35]

            \tikzstyle{every node}=[circle, draw, minimum size=0.5cm, inner sep=0pt]

            \node (B) at (0, -2) [label={[xshift=3pt]left:{\footnotesize $8$}}] {\footnotesize $z_3$};
            \node (C) at (-2, -4) [label={[xshift=3pt]left:{\footnotesize $4$}}] {\footnotesize $z_4$};
            \node (D) at (2, -4) [label=left:{\footnotesize $\frac{5}{2} \ \xcancel{2} \ \bcancel{4}$}]{\footnotesize $x_3'$};
            \node (E) at (-2, -6) [label={[xshift=3pt]left:{\footnotesize $2$}}] {\footnotesize $x_4'$};
            \node (F) at (2, -6) [label=left:{\footnotesize $\frac{5}{4} \ \xcancel{1} \ \bcancel{2}$}]{\footnotesize $x_1'$};
            \node (G) at (-2, -8) [label={[xshift=3pt]left:{\footnotesize $1$}}]{\footnotesize $x_2'$};
            \node (H) at (2, -8) [label=left:{\footnotesize $\frac{1}{2} \ \bcancel{1}$}] {\footnotesize $z_1'$};
            \node (I) at (-2, -10) [label={[xshift=3pt]left:{\footnotesize $\frac{1}{2}$}}] {\footnotesize $z_2'$};
            \node (J) at (4, -6) [label={[xshift=-3pt]right:{\footnotesize $1$}}] {\footnotesize $x_5'$};

            \draw[double] (B) -- (C);
            \draw[double] (B) -- (D);
            \draw[double] (C) -- (E);
            \draw[double] (D) -- (F);
            \draw[double] (E) -- (G);
            \draw[double] (F) -- (H);
            \draw[double] (G) -- (I);
            \draw (D) -- (J);
            
            \end{tikzpicture}
        }
    \end{minipage}
    \hspace{0.095\textwidth}
    \begin{minipage}{0.1\textwidth}
        \centering
        \subfigure[]{
            \begin{tikzpicture}[scale=0.35]

            \tikzstyle{every node}=[circle, draw, minimum size=0.5cm, inner sep=0pt]
            
            \node (B) at (0, -2) [label={[xshift=3pt]left:{\footnotesize $8$}}] {\footnotesize $x_1$};
            \node (C) at (0, -4) [label={[xshift=3pt]left:{\footnotesize $4$}}] {\footnotesize $z_1$};
            \node (D) at (0, -6) [label={[xshift=3pt]left:{\footnotesize $2$}}]{\footnotesize $z_2$};
            \node (E) at (0, -8) [label={[xshift=3pt]left:{\footnotesize $1$}}]{\footnotesize $z_2'$};
            \node (F) at (-1.5, -10) [label={[xshift=3pt]left:{\footnotesize $\frac{1}{2}$}}] {\footnotesize $x_2'$};
            \node (G) at (1.5, -10) [label={[xshift=3pt]left:{\footnotesize $\frac{1}{2}$}}] {\footnotesize $z_1'$};
            \node (H) at (2, -4) [label={[xshift=-3pt]right:{\footnotesize $2$}}] {\footnotesize $x_4$};
            \node (I) at (1.5, -12) [label={[xshift=3pt]left:{\footnotesize $\frac{1}{4}$}}] {\footnotesize $x_1'$};

            \draw[double] (B) -- (C);
            \draw[double] (C) -- (D);
            \draw[double] (D) -- (E);
            \draw[double] (E) -- (F);
            \draw[double] (E) -- (G);
            \draw (B) -- (H);
            \draw (G) -- (I);

            \end{tikzpicture}
        }
    \end{minipage}
    \hspace{0.045\textwidth}
    \begin{minipage}{0.1\textwidth}
        \centering
        \subfigure[]{
            \begin{tikzpicture}[scale=0.35]

             \tikzstyle{every node}=[circle, draw, minimum size=0.5cm, inner sep=0pt]

            \node (B) at (0, -2) [label={[xshift=3pt]left:{\footnotesize $8$}}] {\footnotesize $x_5$};
            \node (C) at (0, -4) [label={[xshift=3pt]left:{\footnotesize $4$}}] {\footnotesize $z_5$};
            \node (D) at (0, -6) [label={[xshift=3pt]left:{\footnotesize $2$}}]{\footnotesize $z_5'$};
            \node (E) at (-1.5, -8) [label={[xshift=3pt]left:{\footnotesize $1$}}] {\footnotesize $x_5'$};
            \node (F) at (1.5, -8) [label={[xshift=3pt]left:{\footnotesize $1$}}] {\footnotesize $z_1'$};
            \node (G) at (-1.5, -10) [label={[xshift=3pt]left:{\footnotesize $\frac{1}{2}$}}] {\footnotesize $x_2'$};
            \node (H) at (1.5, -10) [label={[xshift=3pt]left:{\footnotesize $\frac{1}{2}$}}] {\footnotesize $z_2'$};
            \node (I) at (2, -4) [label={[xshift=-3pt]right:{\footnotesize $2$}}] {\footnotesize $x_2$};
            \node (J) at (3.5, -10) [label={[xshift=-3pt]right:{\footnotesize $\frac{1}{2}$}}] {\footnotesize $x_1'$};
            \node (K) at (-4.5, -4) [draw, dashed, label={[xshift=3pt]left:{\footnotesize $4$}}] {\footnotesize $x_2$};
            \node (L) at (-4.5, -7) [draw, dashed, label={[xshift=3pt]left:{\footnotesize $2$}}] {\footnotesize $z_2$};
            \node (M) at (-4.5, -10) [draw, dashed, label={[xshift=3pt]left:{\footnotesize $1$}}] {\footnotesize $z_2'$};

            \draw[double] (B) -- (C);
            \draw[double] (C) -- (D);
            \draw[double] (D) -- (E);
            \draw[double] (D) -- (F);
            \draw[double] (E) -- (G);
            \draw[double] (F) -- (H);
            \draw (B) -- (I);
            \draw (F) -- (J);
            \draw[double, dashed] (B) -- (K);
            \draw[double, dashed] (K) -- (L);
            \draw[double, dashed] (L) -- (M);
            
            \end{tikzpicture}
        }
    \end{minipage}
    \caption{Application of the heuristic method to the target $x_3$ of (a) Blanuša~2 graph. We have $e(x_3) = 4$, $N(x_3) = \{z_3, x_1, x_5 \}$, and $P(x_3) = \{x_2', z_1', z_2'\}$. The neighbors $z_3$, $x_1$, and $x_5$ of $x_3$ lead to the strategies (b) $T_1$, (c) $T_2$, (d) $T_3$. Certificate $\Tau = \{T_1, T_2, T_3\}$ gives the WFL ratio $\lambda_\Tau = 29.25$. We omit the root $x_3$ in the representation of the subtrees. The edges of the trunk are represented by double lines, while the branches are represented by standard lines. The minimum weight, calculated by Eq.~\eqref{eq:min}, is $\omega_{\text{min}} = 2$ and its minimizer is $P_{\text{min}}(x_3) = \{x_2' \}$. Therefore, we have surplus weight in $z_1'$ and $z_2'$ to work with. For $z_2'$, the shortest path from $x_5$ to $z_2'$, indicated by dashed edges and vertices in subtree $T_3$, is replaced by the non-shortest path $x_5$, $z_5$, $z_5'$, $z_1'$, $z_2'$ in order to reduce the surplus weight. For $z_1'$, on the other hand, we take half of the weight of $z_1'$ in $T_1$, as well as of its two closest ancestors, eliminating the surplus weight on the vertex $x_3'$. The weights are updated with the original value crossed out with a single strikethrough. Observe that we have no surplus weight beyond the theoretical limit of the algorithm. However, our choice for the trunks of $\Tau$ is not the only one with this property. Specifically, in the subtree $T_1$, we could reconfigure the subgraph induced by the vertices $x_1'$ and $z_1'$ by placing $x_1'$ as a child of $x_4'$ and eliminating the vertex $x_3'$. For step $4$, the vertices that have not yet reached the minimum weight are $x_5'$, $x_4$, $x_2$, and $x_1'$. The vertices $x_5'$, $x_4$, and $x_2$ are solved directly by adding branches. For $x_1'$, which already has weight $1$ in the trunk of $T_1$, we can reach a total of $1.75$ by adding branches in $T_2$ and $T_3$. To reach the minimum of $2$ for the vertex $x_1'$, we update the weight of $x_1'$ to $5/4$ in $T_1$, forcing us to create surplus weight in its ancestor $x_3'$. In this case, the weights are updated with the original value crossed out with a double strikethrough.
    }
    \label{fig:B2_x3}
\end{figure}

The goal of steps $3$ and $4$ is to minimize the sum of surplus weights in each vertex of the graph by selecting combinations of shortest paths and branches, respectively. For step $3$, additional strategies can improve the process in some cases. Let $P_{\text{min}}(r)$ be the set of minimizers $\arg\min_{u \in P(r)} \omega_{\text{min}}(u)$ of Eq.~\eqref{eq:min}. Note that if $u \in P_{\text{min}}(r)$, $\omega_{\Tau}(u) = \omega_{\text{min}}$. Otherwise, $\omega_{\Tau}(u) > \omega_{\text{min}}$, and the surplus weight on $u$ gives us two possibilities that can be combined: either simply manually assign less weight to $u$ and, if possible, decrease the weight of ancestors of $u$ in the subtree, or replace the shortest path for one of greater length.

Our method can be partially assisted by a computer, automating tasks such as identifying $r$-peripheral vertices, computing their distances to neighbors, and listing minimum paths. However, in steps $3$ and $4$, we go through a decision-making process to choose shortest paths and branches, respectively. Although we have a clear mathematical criterion, which is to minimize surplus weight, we have not proposed an algorithmic process to deal with this task. In the case of step $3$, if $P(r) = P_{\text{min}}(r)$, we have no surplus weight in the $r$-peripheral vertices and we could solve it by using a brute-force approach to test all combinations of shortest paths. Naturally, the number of combinations is exponential in the worst case. On the other hand, if $P(r) \neq P_{\text{min}}(r)$, then we have different approaches to deal with surplus weight on the $r$-peripheral vertices, making the decision-making process more complicated. In step $4$, we add branches, which is an even less structured process. Thus, our heuristic is feasible for small graphs like Blanuša~2 and well-structured infinite families such as Flower graphs, where analyzing smaller cases helps identify structural patterns to construct the general certificate.

\paragraph{Applications on the snark graphs}

As an initial example, the application of our heuristic approach on the Petersen graph results in the $r$-strategies of Figure~\ref{fig:petersen} up to a multiplication factor of $2$. We have $e(r) = 2$ and each of the six vertices of $P(r)$ is at distance $1$ from one neighbor of $r$ and at distance $2$ from the other two neighbors. Therefore, by Eq.~\eqref{eq:min}, $\omega_{\text{min}} = 2$. In step $3$, all the shortest paths from the neighbors of $r$ to $r$-peripheral vertices are unique. Furthermore, since $P(r) = P_{\text{min}}$, we have no surplus weight on the $r$-peripheral vertices. Step $4$ is not required, since all vertices already reach the minimum weight in step $3$.

Applying the heuristic to the Blanuša graph $B_2$ and the Flower graphs $J_m \ (m = 2k + 1, m \geq 3$), we improved the best-known upper bounds on $\lambda_{G, r}$ from Adauto et al.~\cite{WFL_snarks} for most targets. Note that for $J_3$, Theorem~\ref{thm:min} gives $\lambda_{J_3} \geq 12$, so the dual approach is not able to improve the bound $\pi(J_3) \leq 13$ from Adauto et al.~\cite{WFL_snarks}. As a result, we obtained better bounds for the pebbling number of the graphs in Theorems~\ref{thm:Blan} ($B_2$) and~\ref{thm:Flow} ($J_m, \ m \geq 5$), where we detail the $r$-strategies for a representative target $r$ in each case.

\begin{figure}[!t]
    \centering
    \hspace{-0.125\textwidth}
    \begin{minipage}{0.2\textwidth}
        \centering
        \subfigure[]{
            \begin{tikzpicture}[scale=0.35]

            \tikzstyle{every node}=[circle, draw, minimum size=0.5cm, inner sep=0pt]

            \node (A) at (0, 0) [label={[xshift=3pt]left:{\footnotesize $4$}}] {\footnotesize $v_0$};
            \node (B) at (-2.75, -2) [label={[xshift=-3pt]right:{\footnotesize $2$}}] {\footnotesize $v_{-1}$};
            \node (C) at (2.75, -2) [label={[xshift=3pt]left:{\footnotesize $2$}}] {\footnotesize $v_1$};
            \node (D) at (-2.75, -4) [label={[xshift=-3pt]right:{\footnotesize $1$}}] {\footnotesize $z_{-1}$};
            \node (E) at (2.75, -4) [label={[xshift=3pt]left:{\footnotesize $1$}}] {\footnotesize $z_1$};
            \node (F) at (-4.25, -6) [label={[xshift=3pt]left:{\footnotesize $\frac{1}{2}$}}] {\footnotesize $x_{-1}$};
            \node (G) at (-1.25, -6) [label={[xshift=3pt]left:{\footnotesize $\frac{1}{2}$}}] {\footnotesize $y_{-1}$};
            \node (H) at (1.25, -6) [label={[xshift=-3pt]right:{\footnotesize $\frac{1}{2}$}}] {\footnotesize $x_1$};
            \node (I) at (4.25, -6) [label={[xshift=-3pt]right:{\footnotesize $\frac{1}{2}$}}] {\footnotesize $y_1$};
            
            \draw[double] (A) -- (B);
            \draw[double] (A) -- (C);
            \draw[double] (B) -- (D);
            \draw[double] (C) -- (E);
            \draw (D) -- (F);
            \draw (D) -- (G);
            \draw (E) -- (H);
            \draw (E) -- (I);
            
            \end{tikzpicture}
        }
    \end{minipage}
    \hspace{0.1\textwidth}
    \begin{minipage}{0.2\textwidth}
        \centering
        \subfigure[]{
            \begin{tikzpicture}[scale=0.35]

            \tikzstyle{every node}=[circle, draw, minimum size=0.5cm, inner sep=0pt]
            
            \node (A) at (0, 0) [label={[xshift=3pt]left:{\footnotesize $8$}}] {\footnotesize $v_0$};
            \node (B) at (-2.75, -2) [label={[xshift=-3pt]right:{\footnotesize $4$}}] {\footnotesize $v_{-1}$};
            \node (C) at (2.75, -2) [label={[xshift=3pt]left:{\footnotesize $4$}}] {\footnotesize $v_1$};
            \node (D) at (-2.75, -4) [label={[xshift=-3pt]right:{\footnotesize $2$}}] {\footnotesize $v_{-2}$};
            \node (E) at (2.75, -4) [label={[xshift=3pt]left:{\footnotesize $2$}}] {\footnotesize $v_2$};
            \node (F) at (-2.75, -6) [label={[xshift=-3pt]right:{\footnotesize $1$}}] {\footnotesize $z_{-2}$};
            \node (G) at (2.75, -6) [label={[xshift=3pt]left:{\footnotesize $1$}}] {\footnotesize $z_2$};
            \node (H) at (-4.75, -4) [label={[xshift=3pt]left:{\footnotesize $1$}}] {\footnotesize $z_{-1}$};
            \node (I) at (4.75, -4) [label={[xshift=-3pt]right:{\footnotesize $1$}}] {\footnotesize $z_1$};
            \node (J) at (-4.25, -8) [label={[xshift=3pt]left:{\footnotesize $\frac{1}{2}$}}] {\footnotesize $x_{-2}$};
            \node (K) at (-1.25, -8) [label={[xshift=3pt]left:{\footnotesize $\frac{1}{2}$}}] {\footnotesize $y_{-2}$};
            \node (L) at (1.25, -8) [label={[xshift=-3pt]right:{\footnotesize $\frac{1}{2}$}}] {\footnotesize $x_2$};
            \node (M) at (4.25, -8) [label={[xshift=-3pt]right:{\footnotesize $\frac{1}{2}$}}] {\footnotesize $y_2$};
            
            \draw[double] (A) -- (B);
            \draw[double] (A) -- (C);
            \draw[double] (B) -- (D);
            \draw[double] (C) -- (E);
            \draw[double] (D) -- (F);
            \draw[double] (E) -- (G);
            \draw (B) -- (H);
            \draw (C) -- (I);
            \draw (F) -- (J);
            \draw (F) -- (K);
            \draw (G) -- (L);
            \draw (G) -- (M);

            \end{tikzpicture}
        }
    \end{minipage}
    \hspace{0.125\textwidth}
    \begin{minipage}{0.2\textwidth}
        \centering
        \subfigure[]{
            \begin{tikzpicture}[scale=0.35]

            \tikzstyle{every node}=[circle, draw, minimum size=0.5cm, inner sep=0pt]

            \node (A) at (0, 0) [label=left:{\footnotesize $2^{k + 1}$}] {\footnotesize $v_0$};
            \node (B) at (-2.75, -2) [label={[xshift=-2pt]right:{\footnotesize $2^{k}$}}] {\footnotesize $v_{-1}$};
            \node (C) at (2.75, -2) [label={[xshift=2pt]left:{\footnotesize $2^{k}$}}] {\footnotesize $v_1$};
            \node (D) at (-2.75, -4) [label=right:{\footnotesize $2^{k - 1}$}] {\footnotesize $v_{-2}$};
            \node (E) at (2.75, -4) [label=left:{\footnotesize $2^{k - 1}$}] {\footnotesize $v_2$};
            \node (F) at (-2.75, -9) [label={[xshift=-3pt]right:{\footnotesize $2$}}] {\footnotesize $v_{-k}$};
            \node (G) at (2.75, -9) [label={[xshift=3pt]left:{\footnotesize $2$}}] {\footnotesize $v_k$};
            \node (H) at (-2.75, -11) [label={[xshift=-3pt]right:{\footnotesize $1$}}] {\footnotesize $z_{-k}$};
            \node (I) at (2.75, -11) [label={[xshift=3pt]left:{\footnotesize $1$}}] {\footnotesize $z_k$};
            \node (J) at (-4.75, -4) [label={[xshift=3pt]left:{\footnotesize $1$}}] {\footnotesize $z_{-1}$};
            \node (K) at (4.75, -4) [label={[xshift=-3pt]right:{\footnotesize $1$}}] {\footnotesize $z_1$};
            \node (L) at (-4.75, -6) [label={[xshift=3pt]left:{\footnotesize $1$}}] {\footnotesize $z_{-2}$};
            \node (M) at (4.75, -6) [label={[xshift=-3pt]right:{\footnotesize $1$}}] {\footnotesize $z_2$};
            \node (N) at (-4.75, -9) [label={[xshift=3pt]left:{\footnotesize $1$}}] {\scriptsize $z_{1 - k}$};
            \node (O) at (4.75, -9) [label={[xshift=-3pt]right:{\footnotesize $1$}}] {\scriptsize $z_{k - 1}$};
            \node (P) at (-4.25, -13) [label={[xshift=3pt]left:{\footnotesize $\frac{1}{2}$}}] {\footnotesize $x_{-k}$};
            \node (Q) at (-1.25, -13) [label={[xshift=3pt]left:{\footnotesize $\frac{1}{2}$}}] {\footnotesize $y_{-k}$};
            \node (R) at (1.25, -13) [label={[xshift=-3pt]right:{\footnotesize $\frac{1}{2}$}}] {\footnotesize $x_k$};
            \node (S) at (4.25, -13) [label={[xshift=-3pt]right:{\footnotesize $\frac{1}{2}$}}] {\footnotesize $y_k$};

            \fill (-2.75, -6.25) circle (2pt);
            \fill (2.75, -6.25) circle (2pt);
            \fill (-2.75, -6.5) circle (2pt);
            \fill (2.75, -6.5) circle (2pt);
            \fill (-2.75, -6.75) circle (2pt);
            \fill (2.75, -6.75) circle (2pt);
            
            \draw[double] (A) -- (B);
            \draw[double] (A) -- (C);
            \draw[double] (B) -- (D);
            \draw[double] (C) -- (E);
            \draw[double] (D) -- (-2.75, -5.75);
            \draw[double] (E) -- (2.75, -5.75);
            \draw[double] (-2.75, -7.25) -- (F);
            \draw[double] (2.75, -7.25) -- (G);
            \draw[double] (F) -- (H);
            \draw[double] (G) -- (I);
            \draw (B) -- (J);
            \draw (C) -- (K);
            \draw (D) -- (L);
            \draw (E) -- (M);
            \draw (-3, -7.25) -- (N);
            \draw (3, -7.25) -- (O);
            \draw (H) -- (P);
            \draw (H) -- (Q);
            \draw (I) -- (R);
            \draw (I) -- (S);
            
            \end{tikzpicture}
        }
    \end{minipage}
    \caption{Application of the heuristic method to the target $z_0$ of Flower graph $J_m$. We start by building the strategies of (a) $J_3$ and (b) $J_5$ to arrive at the pattern of the strategies of (c) $J_m$, with $m = 2k + 1, m > 5$. For each graph, we show only one strategy ($T_1$). The other two strategies are obtained by swapping the vertex labels $v$ with $x$ ($T_2$) and $y$ ($T_3$), respectively. Certificate $\Tau = \{T_1, T_2, T_3\}$ gives the WFL ratio $\lambda_\Tau = 2^{k + 2} 3/2 + 2k - 2$ for $m \geq 3$, particularly resulting in a WFL ratio of $12$ and $26$ for the graphs $J_3$ and $J_5$, respectively. Similarly to the case of Blanuša~2, we omit the root $z_0$ in the representation of the subtrees, and we represent the trunk edges and branches by double lines and standard lines, respectively. For any $m \geq 3$, $e(z_0) = k + 2$, $N(z_0) = \{v_0, x_0, y_0 \}$, $P(z_0) = \{z_k, z_{-k} \}$, and the minimum weight, calculated by Eq.~\eqref{eq:min}, is $\omega_{\text{min}} = 3$. Furthermore, $P_{\text{min}}(z_0) = P(z_0)$, and the shortest path in $G - \{ z_0 \}$ between each neighbor of $z_0$ and each $z_0$-peripheral vertex is unique. Therefore, each set of trunks is unique to its respective graph. For the branches added in step $4$, observe that there is always a missing unit in the weight of the vertices $v_k$, $v_{-k}$, $x_k$, $x_{-k}$, $y_k$ and $y_{-k}$. We solve these vertices by adding two branches at each of the vertices $z_k$ and $z_{-k}$, assigning weight $1/2$ to the new vertices. This is enough to finalize the certificate of the $J_3$ graph but observe that vertices $z_1$ and $z_{-1}$ have no weight after step $3$ in graph $J_5$. We solve this case by adding the branches $v_{-1}z_{-1}$ and $v_1z_1$, with weight $1$ on both $z_{-1}$ and $z_1$. In general, the vertices $z_{j}$, $z_{-j}$, for each $1 \leq j \leq k - 1$, have no weight in $J_m$ after step $3$ and to them we add the branches $v_{-j}z_{-j}$ and $v_jz_j$, with weight $1$ on both $z_{-j}$ and $z_j$, finally getting for the $z_0$-strategy shown in (c).
    }
    \label{fig:Jm_z0}
\end{figure}
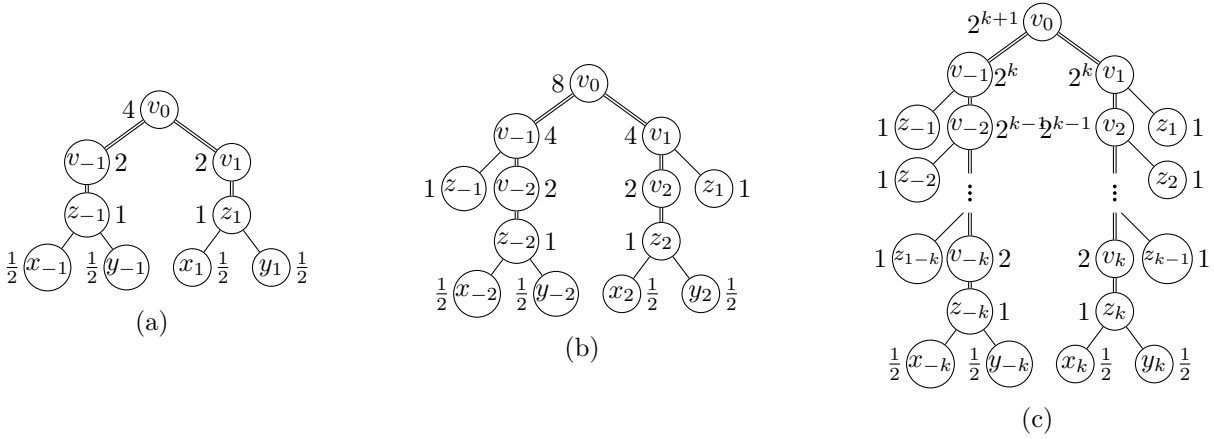

\begin{theorem} \label{thm:Blan}
The pebbling number of the Blanuša graph $B_2$ is bounded by $\pi(B_2) \leq 30$.
\end{theorem}

\begin{proof}
There are six vertex classes, which can be represented, without loss of generality, as the targets $x_1, x_1, x_3, z_1, z_2, z_3$. We choose the target $x_3$ to detail. See Figure~\ref{fig:B2_x3}. The upper bounds on $\lambda_{G, r}$ for all targets are shown in Table~\ref{table:B2}. Observe that $\lambda_{B_2} \leq 29.25$. Therefore, $\pi(B_2) \leq 30$ and the theorem follows. 
\end{proof}

\begin{theorem} \label{thm:Flow}
For all $m \geq 5$, the pebbling number of the Flower graph $J_m$, with $m = 2k + 1$, is bounded by $\pi(J_m) \leq \lfloor 2^{k + 2} 8/5 + 2k - 8/5 \rfloor + 1$. In particular, $\pi(J_5) \leq 29$ and $\pi(J_7) \leq 56$. 
\end{theorem}

\begin{proof}
For any $m$, $J_m$ has rotational symmetry (with a necessary twist), as well as reflective symmetry that negates subscripts, yielding three vertex classes. In particular, we can choose, without loss of generality, $x_0$, $v_0$, and $z_0$ as our target vertices. We choose the target $z_0$ to detail. See Figure~\ref{fig:Jm_z0}. The upper bounds on $\lambda_{G, r}$ for all targets are shown in Table~\ref{table:Jm}. Note that the bottleneck of $J_m \ (m \geq 5)$ is the vertex $v_0$, in which remarkably our heuristic obtained the same result as Adauto et al.~\cite{WFL_snarks}. Thus, $\pi(J_m) \leq \lfloor 2^{k + 2} 8/5 + 2k - 8/5 \rfloor + 1$, as claimed.
\end{proof}

The upper bounds of Adauto et al.~\cite{WFL_snarks} on the pebbling numbers of Blanuša graph $B_2$ and Flower graphs $J_m \ (m = 2k + 1, m \geq 3)$ are, respectively, $34$ and $\lfloor 2^{k + 2} 9/5 + 2k - 18/5 \rfloor + 1$. Therefore, Theorem~\ref{thm:Blan} improves the first in $4$ units, while Theorem~\ref{thm:Flow} improves the last one by a constant factor of $9/8$. Beyond the graph as a whole, Tables~\ref{table:B2} (Blanuša graph) and~\ref{table:Jm} (Flower graphs) make the comparison between our bounds and the bounds of Adauto et al.~\cite{WFL_snarks} with a sharper eye for target vertices, considering the metric of the WFL ratio.  

\begin{table}[ht]
\begin{adjustbox}{width=1\textwidth,center=\textwidth}
    \footnotesize
    \centering
\begin{tabular}{|c|c|c|c|c|c|c|c|}
\hline
\multirow{2}{*}{} & \multicolumn{7}{c|}{\textbf{WFL ratio upper bounds for Blanuša $B_2$}} \\
\cline{2-8}
 & $x_1$ & $x_2$ & $x_3$ & $z_1$ & $z_2$ & $z_3$ & Graph \\
\hline
\textbf{Our heuristic} & $29.25$ & $26.6$ & $29.25$ & $26.6$ & $29$ & $27.2$ & $29.25$ \\
\hline
\textbf{Adauto et al.~\cite{WFL_snarks}} & $30$ & $31$ & $236/7 \approx 33.714$ & $26.6$ & $30$ & $32$ & $236/7 \approx 33.714$ \\
\hline
\end{tabular}
\end{adjustbox}
    \caption{Comparison between our upper bounds for the minimum WFL ratio of the Blanuša graph $B_2$ and those previously established by Adauto et al.~\cite{WFL_snarks}. We have improved the upper bounds for almost all roots except $z_1$. We highlight targets $x_2$, $x_3$, and $z_3$, in which we improve the upper bound on more than $4$ units. Observe that Theorem~\ref{thm:min} gives $\lambda_{B_2, r} \geq 22$ for all targets. Therefore, the gap between the lower and the upper bound for $\lambda_{B_2, r}$ lies between $4.6$ and $7.25$, depending on the target $r$.}
    \label{table:B2}
\end{table}

\begin{table}[ht]
    \centering
    \begin{adjustbox}{width=1\textwidth,center=\textwidth}
    \footnotesize
\begin{tabular}{|c|c|c|c|c|c|c|c|c|c|c|c|c|}
    \hline
    \multirow{3}{*}{} & \multicolumn{8}{c|}{\textbf{WFL ratio upper bounds for Flower snarks}} \\ \cline{2-9}
    & \multicolumn{4}{c|}{\textbf{Our heuristic}} & \multicolumn{4}{c|}{\textbf{Adauto et al.~\cite{WFL_snarks}}} \\ \cline{2-9}
    & $J_3$ & $J_5$ & $J_7$ & $J_m \ (m = 2k + 1, m > 7)$ & $J_3$ & $J_5$ & $J_7$ & $J_m \ (m = 2k + 1, m > 7)$ \\ \hline
    \textbf{$x_0$} & $12.8$ & $27.8$ & $54.6$ & $2^{k + 2}  13/10 + 6k - 5$ & $12.8$ & $29.2$ & $56.8$ & $2^{k + 2}  17/10 + 2k - 18/5$ \\ \hline
    \textbf{$v_0$} & $12.8$ & $28$ & $55.6$ & $2^{k + 2}  8/5 + 2k - 8/5$ & $12.8$ & $29.2$ & $55.6$ & $2^{k + 2}  8/5 + 2k - 8/5$ \\ \hline
     \textbf{$z_0$} & $12$ & $26$ & $52$ & $2^{k + 2} 3/2 + 2k - 2$ & $12.8$ & $29.2$ & $60$ & $2^{k + 2}  9/5 + 2k - 18/5$ \\ \hline
    \textbf{Graph} & $12.8$ & $28$ & $55.6$ & $2^{k + 2} 8/5 + 2k - 8/5$ & $12.8$ & $29.2$ & $60$ & $2^{k + 2} 9/5 + 2k - 18/5$ \\ \hline
\end{tabular}
\end{adjustbox}
    \caption{Comparison between our upper bounds for the minimum WFL ratio of the Flower graphs $J_m \ (m = 2k + 1, m \geq 3)$ and those previously established by Adauto et al.~\cite{WFL_snarks}. For $J_3$, we improve the minimum WFL ratio bound of the target $z_0$ from $12.8$ to $12$, which we conclude by Theorem~\ref{thm:min} that is optimal. For $m > 7$, we have improved the $J_m$ upper bounds for the targets $x_0$ and $z_0$ asymptotically by the constant factors of $17/13$ and $6/5$, respectively. For $v_0$, we obtained the same result as Adauto et al.~\cite{WFL_snarks}, which in particular is the bottleneck for $\lambda_{J_m}$ in our results. On the other hand, we can conclude using Eq.~\eqref{eq:alt} that Theorem~\ref{thm:min} gives $\lambda_{J_m, r} \geq 2^{k + 2} + 10 \ (m \geq 7)$ 
    for all targets, which implies the gap from our heuristic to our theoretical lower bound lies asymptotically in a constant factor of $13/10$ and $8/5$, depending on the target.}
    \label{table:Jm}
\end{table}

\paragraph{Final remarks}

Using our heuristic method, we have narrowed the known window for the pebbling number of the Blanuša graph $B_2$ and the Flower graph infinite family $J_m \ (m = 2k + 1, m \geq 5)$. Our method is grounded in a theoretical lower bound analysis, showing how a deeper understanding of the limitations of the Weight Function Lemma leads to a more structured and effective strategy selection---enhancing both the theoretical insight and the practical applicability of our work. Nonetheless, for the studied graphs, the gap between the lower bounds obtained by Theorem~\ref{thm:min} and the certificates obtained by the heuristic method---especially for the Flower snarks---is still considerable, indicating potential for improvement in both directions.

Regarding Theorem~\ref{thm:min}, the tight bound requires that, for each tuple $(j, v_j, u)$, where $j \in I_{\text{sur}}$, $v_j$ is a vertex in the $j^{\text{th}}$ neighborhood, and $u$ is a $r$-peripheral vertex, there exists a shortest path from $r$ to $u$ that passes through $v_j$. However, the greater the actual distance between the vertices of the $j^{\text{th}}$ neighborhood and the $r$-peripheral vertices, the harder it becomes to counterbalance the surplus weight caused by the $r$-peripheral vertices. Future refinements could take into account, for each neighborhood, a function of the distances between its vertices and the $r$-peripheral vertices. For the particular case of the first neighborhood, this function could correspond to Eq.~\eqref{eq:min}. 

On the heuristic side, our approach was designed based on the particular case of $I_{\text{sur}} = \{ 1\} $ of Theorem~\ref{thm:min}, i.e., by building the subtrees from paths connecting the vertices of the first neighborhood and the $r$-peripheral vertices. Future work may explore paths originating from other neighborhoods $N_j(r)$, with the choice of $j$ guided by a criterion that considers the potential surplus weight introduced by the scheme based on each $N_j(r)$. In such an approach, one must account not only for the paths from the chosen neighborhood to the $r$-peripheral vertices but also for the paths from those vertices back to the root of each subtree. Alternatively, more sophisticated methods could aim to combine paths from multiple neighborhoods simultaneously, seeking a more globally balanced surplus weight distribution. In this context, it becomes increasingly important to automate through some algorithmic procedure the steps $3$ and $4$ of the heuristic method, since similar routines are likely to be required in the development of more refined heuristics.

\bibliographystyle{plain}

\begin{thebibliography}{10}
\expandafter\ifx\csname url\endcsname\relax
  \def\url#1{\texttt{#1}}\fi
\expandafter\ifx\csname urlprefix\endcsname\relax\def\urlprefix{URL }\fi
\expandafter\ifx\csname href\endcsname\relax
  \def\href#1#2{#2} \def\path#1{#1}\fi

\bibitem{recent_survey}
G.~Hurlbert, F.~Kenter, Graph pebbling: A blend of graph theory, number theory,
  and optimization, Notices of the American Mathematical Society 68~(11) (1900)
  1900--1913.

\bibitem{pebbling_completeness2}
K.~Milans, B.~Clark, The complexity of graph pebbling, SIAM Journal on Discrete
  Mathematics 20~(3) (2006) 769--798.

\bibitem{WFL}
G.~Hurlbert, The weight function lemma for graph pebbling, Journal of
  Combinatorial Optimization 34~(2) (2017) 343--361.

\bibitem{ILP}
C.~H. Papadimitriou, K.~Steiglitz, Combinatorial Optimization: Algorithms and
  Complexity, Courier Corporation, North Chelmsford, 1998.

\bibitem{WFL_snarks}
M.~Adauto, C.~de~Figueiredo, G.~Hurlbert, D.~Sasaki, On the pebbling numbers of
  {F}lower, {B}lanu{\v{s}}a and {W}atkins snarks, Discrete Applied Mathematics
  361 (2025) 336--346.

\bibitem{WFL_kneser}
M.~Adauto, V.~Bardenova, M.~da~Cruz, C.~de~Figueiredo, G.~Hurlbert, D.~Sasaki,
  Pebbling in {K}neser graphs, in: Latin American Symposium on Theoretical
  Informatics, Lecture Notes in Computer Science, Vol. 14579, Springer, 2024,
  pp. 46--60.

\bibitem{pebbling_split}
L.~Alc{\'o}n, M.~Gutierrez, G.~Hurlbert, Pebbling in split graphs, SIAM Journal
  on Discrete Mathematics 28~(3) (2014) 1449--1466.

\bibitem{survey_snarks}
A.~Cavicchioli, M.~Meschiari, B.~Ruini, F.~Spaggiari, A survey on snarks and
  new results: Products, reducibility and a computer search, Journal of Graph
  Theory 28~(2) (1998) 57--86.

\bibitem{fourColors}
P.~G. Tait, Remarks on the colouring of maps, Proceedings of the Royal Society
  of Edinburgh 10 (1880) 729.

\bibitem{petersen}
J.~Petersen, Sur le théorème de {T}ait, L'Intermédiaire des Mathématiciens
  15 (1898) 225--227.

\bibitem{Campos}
C.~Campos, S.~Dantas, C.~de~Mello, The total-chromatic number of some families
  of snarks, Discrete Mathematics 311~(12) (2011) 984--988.

\bibitem{pebbling_completeness1}
N.~G. Watson, The complexity of pebbling and cover pebbling, available as arXiv
  preprint math/0503511 21 Apr 2005.

\bibitem{survey_pebbling}
G.~Hurlbert, A survey of graph pebbling, Congressus Numerantium 139 (1999)
  41--64.

\bibitem{basic_bound}
M.~Chan, A.~P. Godbole, Improved pebbling bounds, Discrete Mathematics 308~(11)
  (2008) 2301--2306.

\end{thebibliography}

\newpage

\appendix 

\section{Flower graphs representations}

In Figure~\ref{fig:Flowers}, we present a visual representation for the Flower graphs $J_3$ and $J_5$.

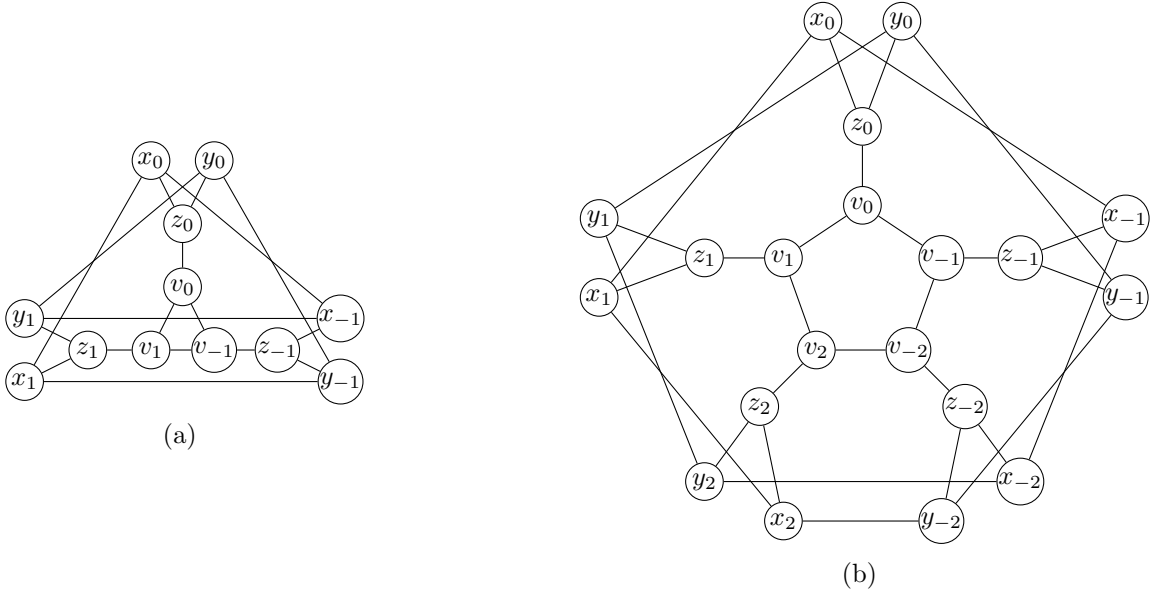
\begin{figure}[ht]
    \centering
    \hspace{-0.125\textwidth}
    \begin{minipage}{0.4\textwidth}
        \centering
        \subfigure[]{
            \begin{tikzpicture}[scale=0.35]

            \tikzstyle{every node}=[circle, draw, minimum size=0.5cm, inner sep=0pt]
            
            \node (A) at (0, 12/5) {\footnotesize $z_0$};
            \node (B) at (0, 0) {\footnotesize $v_0$};
            \node (C) at (-6/5, 24/5) {\footnotesize $x_0$};
            \node (D) at (6/5, 24/5) {\footnotesize $y_0$};
            \node (E) at (-6/5, -12/5) {\footnotesize $v_1$};
            \node (F) at (6/5, -12/5) {\footnotesize $v_{-1}$};
            \node (G) at (-6, -18/5) {\footnotesize $x_1$};
            \node (H) at (6, -6/5) {\footnotesize $x_{-1}$};
            \node (I) at (-6, -6/5) {\footnotesize $y_1$};
            \node (J) at (6, -18/5) {\footnotesize $y_{-1}$};
            \node (K) at (-18/5, -12/5) {\footnotesize $z_1$};
            \node (L) at (18/5, -12/5) {\footnotesize $z_{-1}$};

            \draw (A) -- (B);
            \draw (A) -- (C);
            \draw (A) -- (D);
            \draw (B) -- (E);
            \draw (B) -- (F);
            \draw (C) -- (G);
            \draw (C) -- (H);
            \draw (D) -- (I);
            \draw (D) -- (J);
            \draw (E) -- (K);
            \draw (F) -- (L);
            \draw (E) -- (F);
            \draw (G) -- (K);
            \draw (I) -- (K);
            \draw (H) -- (L);
            \draw (J) -- (L);
            \draw (H) -- (I);
            \draw (G) -- (J);
            
            \end{tikzpicture}
        }
    \end{minipage}
    \hspace{0.075\textwidth}
    \begin{minipage}{0.4\textwidth}
        \centering
        \subfigure[]{
            \begin{tikzpicture}[scale=0.35]

            \tikzstyle{every node}=[circle, draw, minimum size=0.5cm, inner sep=0pt]
            \node (A) at (0, 3) {\footnotesize $v_0$};
            \node (B) at (-3, 1) {\footnotesize $v_{1}$};
            \node (C) at (3, 1) {\footnotesize $v_{-1}$};
            \node (D) at (-1.75, -2.5) {\footnotesize $v_{2}$};
            \node (E) at (1.75, -2.5) {\footnotesize $v_{-2}$};
            \node (F) at (0, 6) {\footnotesize $z_0$};
            \node (G) at (-6, 1) {\footnotesize $z_{1}$};
            \node (H) at (6, 1) {\footnotesize $z_{-1}$};
            \node (I) at (-3.9, -4.65) {\footnotesize $z_{2}$};
            \node (J) at (3.9, -4.65) {\footnotesize $z_{-2}$};
            \node (K) at (-1.5, 10) {\footnotesize $x_0$};
            \node (L) at (1.5, 10) {\footnotesize $y_0$};
            \node (M) at (-10, -0.5) {\footnotesize $x_{1}$};
            \node (N) at (-10, 2.5) {\footnotesize $y_{1}$};
            \node (O) at (10, 2.5) {\footnotesize $x_{-1}$};
            \node (P) at (10, -0.5) {\footnotesize $y_{-1}$};
            \node (Q) at (-3, -9) {\footnotesize $x_{2}$};
            \node (R) at (-6, -7.5) {\footnotesize $y_{2}$};
            \node (S) at (6, -7.5) {\footnotesize $x_{-2}$};
            \node (T) at (3, -9) {\footnotesize $y_{-2}$};

            \draw (A) -- (B);
            \draw (A) -- (C);
            \draw (B) -- (D);
            \draw (C) -- (E);
            \draw (D) -- (E);
            \draw (A) -- (F);
            \draw (B) -- (G);
            \draw (C) -- (H);
            \draw (D) -- (I);
            \draw (E) -- (J);
            \draw (F) -- (K);
            \draw (F) -- (L);
            \draw (G) -- (M);
            \draw (G) -- (N);
            \draw (H) -- (O);
            \draw (H) -- (P);
            \draw (I) -- (Q);
            \draw (I) -- (R);
            \draw (J) -- (S);
            \draw (J) -- (T);
            \draw (K) -- (M);
            \draw (K) -- (O);
            \draw (L) -- (N);
            \draw (L) -- (P);
            \draw (M) -- (Q);
            \draw (N) -- (R);
            \draw (O) -- (S);
            \draw (P) -- (T);
            \draw (Q) -- (T);
            \draw (R) -- (S);

            \end{tikzpicture}
        }
    \end{minipage}
    \caption{Flower graphs (a) $J_3$ and (b) $J_5$.}
    \label{fig:Flowers}
\end{figure}

\section{Complete set of $z_0$-strategies on Flower graphs}

In Figures~\ref{fig:J3_complete},~\ref{fig:J5_complete}, and~\ref{fig:Jm_complete}, we present visual representations for all $z_0$-strategies $\Tau = \{ T_1, T_2, T_3 \}$ of the Flower graphs $J_3$, $J_5$, and $J_m \ (m = 2k + 1, m > 5)$, respectively. These figures complement Figure~\ref{fig:Jm_z0}, where only $T_1$ is shown explicitly.

\begin{figure}[ht]
    \centering
    \hspace{-0.075\textwidth}
    \begin{minipage}{0.2\textwidth}
        \centering
        \subfigure[]{
            \begin{tikzpicture}[scale=0.35]

            \tikzstyle{every node}=[circle, draw, minimum size=0.5cm, inner sep=0pt]

            \node (A) at (0, 0) [label={[xshift=3pt]left:{\footnotesize $4$}}] {\footnotesize $v_0$};
            \node (B) at (-2.75, -2) [label={[xshift=-3pt]right:{\footnotesize $2$}}] {\footnotesize $v_{-1}$};
            \node (C) at (2.75, -2) [label={[xshift=3pt]left:{\footnotesize $2$}}] {\footnotesize $v_1$};
            \node (D) at (-2.75, -4) [label={[xshift=-3pt]right:{\footnotesize $1$}}] {\footnotesize $z_{-1}$};
            \node (E) at (2.75, -4) [label={[xshift=3pt]left:{\footnotesize $1$}}] {\footnotesize $z_1$};
            \node (F) at (-4.25, -6) [label={[xshift=3pt]left:{\footnotesize $\frac{1}{2}$}}] {\footnotesize $x_{-1}$};
            \node (G) at (-1.25, -6) [label={[xshift=3pt]left:{\footnotesize $\frac{1}{2}$}}] {\footnotesize $y_{-1}$};
            \node (H) at (1.25, -6) [label={[xshift=-3pt]right:{\footnotesize $\frac{1}{2}$}}] {\footnotesize $x_1$};
            \node (I) at (4.25, -6) [label={[xshift=-3pt]right:{\footnotesize $\frac{1}{2}$}}] {\footnotesize $y_1$};
            
            \draw[double] (A) -- (B);
            \draw[double] (A) -- (C);
            \draw[double] (B) -- (D);
            \draw[double] (C) -- (E);
            \draw (D) -- (F);
            \draw (D) -- (G);
            \draw (E) -- (H);
            \draw (E) -- (I);
            
            \end{tikzpicture}
        }
    \end{minipage}
    \hspace{0.125\textwidth}
    \begin{minipage}{0.2\textwidth}
        \centering
        \subfigure[]{
            \begin{tikzpicture}[scale=0.35]

            \tikzstyle{every node}=[circle, draw, minimum size=0.5cm, inner sep=0pt]

            \node (A) at (0, 0) [label={[xshift=3pt]left:{\footnotesize $4$}}] {\footnotesize $x_0$};
            \node (B) at (-2.75, -2) [label={[xshift=-3pt]right:{\footnotesize $2$}}] {\footnotesize $x_{-1}$};
            \node (C) at (2.75, -2) [label={[xshift=3pt]left:{\footnotesize $2$}}] {\footnotesize $x_1$};
            \node (D) at (-2.75, -4) [label={[xshift=-3pt]right:{\footnotesize $1$}}] {\footnotesize $z_{-1}$};
            \node (E) at (2.75, -4) [label={[xshift=3pt]left:{\footnotesize $1$}}] {\footnotesize $z_1$};
            \node (F) at (-4.25, -6) [label={[xshift=3pt]left:{\footnotesize $\frac{1}{2}$}}] {\footnotesize $v_{-1}$};
            \node (G) at (-1.25, -6) [label={[xshift=3pt]left:{\footnotesize $\frac{1}{2}$}}] {\footnotesize $y_{-1}$};
            \node (H) at (1.25, -6) [label={[xshift=-3pt]right:{\footnotesize $\frac{1}{2}$}}] {\footnotesize $v_1$};
            \node (I) at (4.25, -6) [label={[xshift=-3pt]right:{\footnotesize $\frac{1}{2}$}}] {\footnotesize $y_1$};
            
            \draw[double] (A) -- (B);
            \draw[double] (A) -- (C);
            \draw[double] (B) -- (D);
            \draw[double] (C) -- (E);
            \draw (D) -- (F);
            \draw (D) -- (G);
            \draw (E) -- (H);
            \draw (E) -- (I);

            \end{tikzpicture}
        }
    \end{minipage}
    \hspace{0.125\textwidth}
    \begin{minipage}{0.2\textwidth}
        \centering
        \subfigure[]{
            \begin{tikzpicture}[scale=0.35]

            \tikzstyle{every node}=[circle, draw, minimum size=0.5cm, inner sep=0pt]

            \node (A) at (0, 0) [label={[xshift=3pt]left:{\footnotesize $4$}}] {\footnotesize $y_0$};
            \node (B) at (-2.75, -2) [label={[xshift=-3pt]right:{\footnotesize $2$}}] {\footnotesize $y_{-1}$};
            \node (C) at (2.75, -2) [label={[xshift=3pt]left:{\footnotesize $2$}}] {\footnotesize $y_1$};
            \node (D) at (-2.75, -4) [label={[xshift=-3pt]right:{\footnotesize $1$}}] {\footnotesize $z_{-1}$};
            \node (E) at (2.75, -4) [label={[xshift=3pt]left:{\footnotesize $1$}}] {\footnotesize $z_1$};
            \node (F) at (-4.25, -6) [label={[xshift=3pt]left:{\footnotesize $\frac{1}{2}$}}] {\footnotesize $x_{-1}$};
            \node (G) at (-1.25, -6) [label={[xshift=3pt]left:{\footnotesize $\frac{1}{2}$}}] {\footnotesize $v_{-1}$};
            \node (H) at (1.25, -6) [label={[xshift=-3pt]right:{\footnotesize $\frac{1}{2}$}}] {\footnotesize $x_1$};
            \node (I) at (4.25, -6) [label={[xshift=-3pt]right:{\footnotesize $\frac{1}{2}$}}] {\footnotesize $v_1$};
            
            \draw[double] (A) -- (B);
            \draw[double] (A) -- (C);
            \draw[double] (B) -- (D);
            \draw[double] (C) -- (E);
            \draw (D) -- (F);
            \draw (D) -- (G);
            \draw (E) -- (H);
            \draw (E) -- (I);
            
            \end{tikzpicture}
        }
    \end{minipage}
    \caption{Set $\Tau = \{T_1, T_2, T_3 \}$ of $z_0$-strategies for the Flower graph $J_3$, where (a) shows $T_1$, (b) $T_2$, and (c) $T_3$. We omit the root $z_0$ in the representation of the subtrees, and we represent the trunk edges and branches by double lines and standard lines, respectively.}
    \label{fig:J3_complete}
\end{figure}
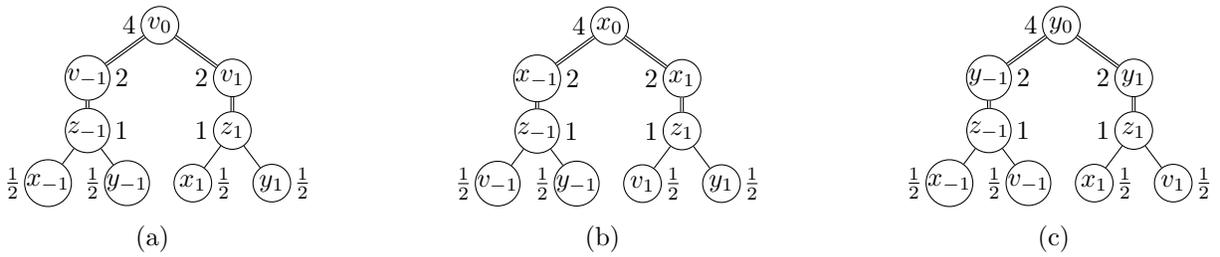

\begin{figure}[ht]
    \centering
    \hspace{-0.1\textwidth}
    \begin{minipage}{0.2\textwidth}
        \centering
        \subfigure[]{
            \begin{tikzpicture}[scale=0.35]

            \tikzstyle{every node}=[circle, draw, minimum size=0.5cm, inner sep=0pt]
            
            \node (A) at (0, 0) [label={[xshift=3pt]left:{\footnotesize $8$}}] {\footnotesize $v_0$};
            \node (B) at (-2.75, -2) [label={[xshift=-3pt]right:{\footnotesize $4$}}] {\footnotesize $v_{-1}$};
            \node (C) at (2.75, -2) [label={[xshift=3pt]left:{\footnotesize $4$}}] {\footnotesize $v_1$};
            \node (D) at (-2.75, -4) [label={[xshift=-3pt]right:{\footnotesize $2$}}] {\footnotesize $v_{-2}$};
            \node (E) at (2.75, -4) [label={[xshift=3pt]left:{\footnotesize $2$}}] {\footnotesize $v_2$};
            \node (F) at (-2.75, -6) [label={[xshift=-3pt]right:{\footnotesize $1$}}] {\footnotesize $z_{-2}$};
            \node (G) at (2.75, -6) [label={[xshift=3pt]left:{\footnotesize $1$}}] {\footnotesize $z_2$};
            \node (H) at (-4.75, -4) [label={[xshift=3pt]left:{\footnotesize $1$}}] {\footnotesize $z_{-1}$};
            \node (I) at (4.75, -4) [label={[xshift=-3pt]right:{\footnotesize $1$}}] {\footnotesize $z_1$};
            \node (J) at (-4.25, -8) [label={[xshift=3pt]left:{\footnotesize $\frac{1}{2}$}}] {\footnotesize $x_{-2}$};
            \node (K) at (-1.25, -8) [label={[xshift=3pt]left:{\footnotesize $\frac{1}{2}$}}] {\footnotesize $y_{-2}$};
            \node (L) at (1.25, -8) [label={[xshift=-3pt]right:{\footnotesize $\frac{1}{2}$}}] {\footnotesize $x_2$};
            \node (M) at (4.25, -8) [label={[xshift=-3pt]right:{\footnotesize $\frac{1}{2}$}}] {\footnotesize $y_2$};
            
            \draw[double] (A) -- (B);
            \draw[double] (A) -- (C);
            \draw[double] (B) -- (D);
            \draw[double] (C) -- (E);
            \draw[double] (D) -- (F);
            \draw[double] (E) -- (G);
            \draw (B) -- (H);
            \draw (C) -- (I);
            \draw (F) -- (J);
            \draw (F) -- (K);
            \draw (G) -- (L);
            \draw (G) -- (M);
            
            \end{tikzpicture}
        }
    \end{minipage}
    \hspace{0.125\textwidth}
    \begin{minipage}{0.2\textwidth}
        \centering
        \subfigure[]{
            \begin{tikzpicture}[scale=0.35]

            \tikzstyle{every node}=[circle, draw, minimum size=0.5cm, inner sep=0pt]
            
            \node (A) at (0, 0) [label={[xshift=3pt]left:{\footnotesize $8$}}] {\footnotesize $x_0$};
            \node (B) at (-2.75, -2) [label={[xshift=-3pt]right:{\footnotesize $4$}}] {\footnotesize $x_{-1}$};
            \node (C) at (2.75, -2) [label={[xshift=3pt]left:{\footnotesize $4$}}] {\footnotesize $x_1$};
            \node (D) at (-2.75, -4) [label={[xshift=-3pt]right:{\footnotesize $2$}}] {\footnotesize $x_{-2}$};
            \node (E) at (2.75, -4) [label={[xshift=3pt]left:{\footnotesize $2$}}] {\footnotesize $x_2$};
            \node (F) at (-2.75, -6) [label={[xshift=-3pt]right:{\footnotesize $1$}}] {\footnotesize $z_{-2}$};
            \node (G) at (2.75, -6) [label={[xshift=3pt]left:{\footnotesize $1$}}] {\footnotesize $z_2$};
            \node (H) at (-4.75, -4) [label={[xshift=3pt]left:{\footnotesize $1$}}] {\footnotesize $z_{-1}$};
            \node (I) at (4.75, -4) [label={[xshift=-3pt]right:{\footnotesize $1$}}] {\footnotesize $z_1$};
            \node (J) at (-4.25, -8) [label={[xshift=3pt]left:{\footnotesize $\frac{1}{2}$}}] {\footnotesize $v_{-2}$};
            \node (K) at (-1.25, -8) [label={[xshift=3pt]left:{\footnotesize $\frac{1}{2}$}}] {\footnotesize $y_{-2}$};
            \node (L) at (1.25, -8) [label={[xshift=-3pt]right:{\footnotesize $\frac{1}{2}$}}] {\footnotesize $v_2$};
            \node (M) at (4.25, -8) [label={[xshift=-3pt]right:{\footnotesize $\frac{1}{2}$}}] {\footnotesize $y_2$};
            
            \draw[double] (A) -- (B);
            \draw[double] (A) -- (C);
            \draw[double] (B) -- (D);
            \draw[double] (C) -- (E);
            \draw[double] (D) -- (F);
            \draw[double] (E) -- (G);
            \draw (B) -- (H);
            \draw (C) -- (I);
            \draw (F) -- (J);
            \draw (F) -- (K);
            \draw (G) -- (L);
            \draw (G) -- (M);

            \end{tikzpicture}
        }
    \end{minipage}
    \hspace{0.125\textwidth}
    \begin{minipage}{0.2\textwidth}
        \centering
        \subfigure[]{
            \begin{tikzpicture}[scale=0.35]

            \tikzstyle{every node}=[circle, draw, minimum size=0.5cm, inner sep=0pt]
            
            \node (A) at (0, 0) [label={[xshift=3pt]left:{\footnotesize $8$}}] {\footnotesize $y_0$};
            \node (B) at (-2.75, -2) [label={[xshift=-3pt]right:{\footnotesize $4$}}] {\footnotesize $y_{-1}$};
            \node (C) at (2.75, -2) [label={[xshift=3pt]left:{\footnotesize $4$}}] {\footnotesize $y_1$};
            \node (D) at (-2.75, -4) [label={[xshift=-3pt]right:{\footnotesize $2$}}] {\footnotesize $y_{-2}$};
            \node (E) at (2.75, -4) [label={[xshift=3pt]left:{\footnotesize $2$}}] {\footnotesize $y_2$};
            \node (F) at (-2.75, -6) [label={[xshift=-3pt]right:{\footnotesize $1$}}] {\footnotesize $z_{-2}$};
            \node (G) at (2.75, -6) [label={[xshift=3pt]left:{\footnotesize $1$}}] {\footnotesize $z_2$};
            \node (H) at (-4.75, -4) [label={[xshift=3pt]left:{\footnotesize $1$}}] {\footnotesize $z_{-1}$};
            \node (I) at (4.75, -4) [label={[xshift=-3pt]right:{\footnotesize $1$}}] {\footnotesize $z_1$};
            \node (J) at (-4.25, -8) [label={[xshift=3pt]left:{\footnotesize $\frac{1}{2}$}}] {\footnotesize $x_{-2}$};
            \node (K) at (-1.25, -8) [label={[xshift=3pt]left:{\footnotesize $\frac{1}{2}$}}] {\footnotesize $v_{-2}$};
            \node (L) at (1.25, -8) [label={[xshift=-3pt]right:{\footnotesize $\frac{1}{2}$}}] {\footnotesize $x_2$};
            \node (M) at (4.25, -8) [label={[xshift=-3pt]right:{\footnotesize $\frac{1}{2}$}}] {\footnotesize $v_2$};
            
            \draw[double] (A) -- (B);
            \draw[double] (A) -- (C);
            \draw[double] (B) -- (D);
            \draw[double] (C) -- (E);
            \draw[double] (D) -- (F);
            \draw[double] (E) -- (G);
            \draw (B) -- (H);
            \draw (C) -- (I);
            \draw (F) -- (J);
            \draw (F) -- (K);
            \draw (G) -- (L);
            \draw (G) -- (M);
            
            \end{tikzpicture}
        }
    \end{minipage}
    \caption{Set $\Tau = \{T_1, T_2, T_3 \}$ of $z_0$-strategies for the Flower graph $J_5$, where (a) shows $T_1$, (b) $T_2$, and (c) $T_3$. We omit the root $z_0$ in the representation of the subtrees, and we represent the trunk edges and branches by double lines and standard lines, respectively.}
    \label{fig:J5_complete}
\end{figure}
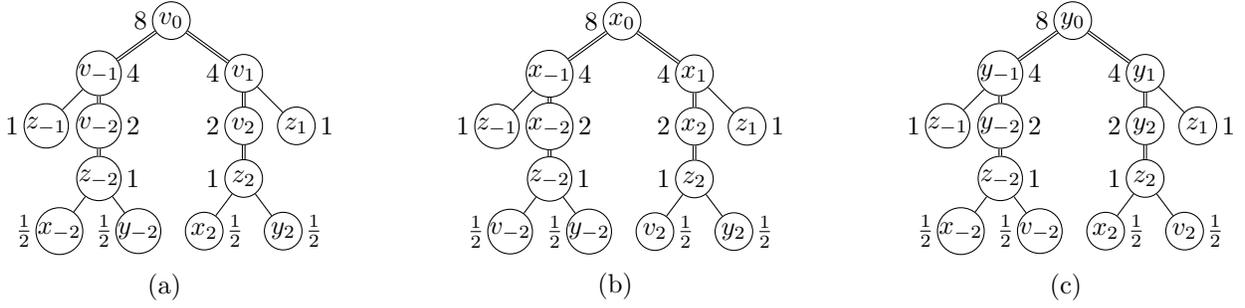

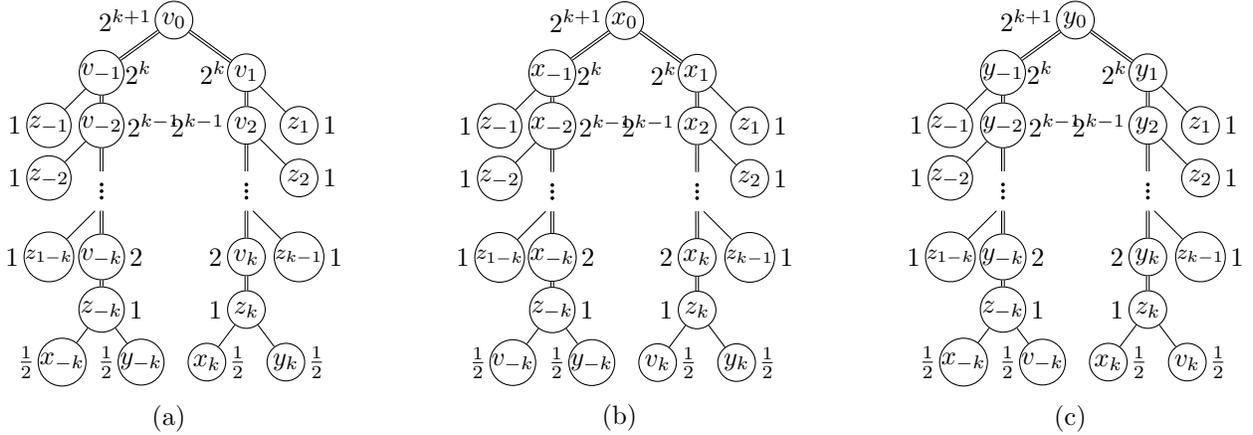
\begin{figure}[ht]
    \centering
    \hspace{-0.1\textwidth}
    \begin{minipage}{0.2\textwidth}
        \centering
        \subfigure[]{
            \begin{tikzpicture}[scale=0.35]

            \tikzstyle{every node}=[circle, draw, minimum size=0.5cm, inner sep=0pt]

            \node (A) at (0, 0) [label=left:{\footnotesize $2^{k + 1}$}] {\footnotesize $v_0$};
            \node (B) at (-2.75, -2) [label={[xshift=-2pt]right:{\footnotesize $2^{k}$}}] {\footnotesize $v_{-1}$};
            \node (C) at (2.75, -2) [label={[xshift=2pt]left:{\footnotesize $2^{k}$}}] {\footnotesize $v_1$};
            \node (D) at (-2.75, -4) [label=right:{\footnotesize $2^{k - 1}$}] {\footnotesize $v_{-2}$};
            \node (E) at (2.75, -4) [label=left:{\footnotesize $2^{k - 1}$}] {\footnotesize $v_2$};
            \node (F) at (-2.75, -9) [label={[xshift=-3pt]right:{\footnotesize $2$}}] {\footnotesize $v_{-k}$};
            \node (G) at (2.75, -9) [label={[xshift=3pt]left:{\footnotesize $2$}}] {\footnotesize $v_k$};
            \node (H) at (-2.75, -11) [label={[xshift=-3pt]right:{\footnotesize $1$}}] {\footnotesize $z_{-k}$};
            \node (I) at (2.75, -11) [label={[xshift=3pt]left:{\footnotesize $1$}}] {\footnotesize $z_k$};
            \node (J) at (-4.75, -4) [label={[xshift=3pt]left:{\footnotesize $1$}}] {\footnotesize $z_{-1}$};
            \node (K) at (4.75, -4) [label={[xshift=-3pt]right:{\footnotesize $1$}}] {\footnotesize $z_1$};
            \node (L) at (-4.75, -6) [label={[xshift=3pt]left:{\footnotesize $1$}}] {\footnotesize $z_{-2}$};
            \node (M) at (4.75, -6) [label={[xshift=-3pt]right:{\footnotesize $1$}}] {\footnotesize $z_2$};
            \node (N) at (-4.75, -9) [label={[xshift=3pt]left:{\footnotesize $1$}}] {\scriptsize $z_{1 - k}$};
            \node (O) at (4.75, -9) [label={[xshift=-3pt]right:{\footnotesize $1$}}] {\scriptsize $z_{k - 1}$};
            \node (P) at (-4.25, -13) [label={[xshift=3pt]left:{\footnotesize $\frac{1}{2}$}}] {\footnotesize $x_{-k}$};
            \node (Q) at (-1.25, -13) [label={[xshift=3pt]left:{\footnotesize $\frac{1}{2}$}}] {\footnotesize $y_{-k}$};
            \node (R) at (1.25, -13) [label={[xshift=-3pt]right:{\footnotesize $\frac{1}{2}$}}] {\footnotesize $x_k$};
            \node (S) at (4.25, -13) [label={[xshift=-3pt]right:{\footnotesize $\frac{1}{2}$}}] {\footnotesize $y_k$};

            \fill (-2.75, -6.25) circle (2pt);
            \fill (2.75, -6.25) circle (2pt);
            \fill (-2.75, -6.5) circle (2pt);
            \fill (2.75, -6.5) circle (2pt);
            \fill (-2.75, -6.75) circle (2pt);
            \fill (2.75, -6.75) circle (2pt);
            
            \draw[double] (A) -- (B);
            \draw[double] (A) -- (C);
            \draw[double] (B) -- (D);
            \draw[double] (C) -- (E);
            \draw[double] (D) -- (-2.75, -5.75);
            \draw[double] (E) -- (2.75, -5.75);
            \draw[double] (-2.75, -7.25) -- (F);
            \draw[double] (2.75, -7.25) -- (G);
            \draw[double] (F) -- (H);
            \draw[double] (G) -- (I);
            \draw (B) -- (J);
            \draw (C) -- (K);
            \draw (D) -- (L);
            \draw (E) -- (M);
            \draw (-3, -7.25) -- (N);
            \draw (3, -7.25) -- (O);
            \draw (H) -- (P);
            \draw (H) -- (Q);
            \draw (I) -- (R);
            \draw (I) -- (S);
            
            \end{tikzpicture}
        }
    \end{minipage}
    \hspace{0.125\textwidth}
    \begin{minipage}{0.2\textwidth}
        \centering
        \subfigure[]{
            \begin{tikzpicture}[scale=0.35]

            \tikzstyle{every node}=[circle, draw, minimum size=0.5cm, inner sep=0pt]

            \node (A) at (0, 0) [label=left:{\footnotesize $2^{k + 1}$}] {\footnotesize $x_0$};
            \node (B) at (-2.75, -2) [label={[xshift=-2pt]right:{\footnotesize $2^{k}$}}] {\footnotesize $x_{-1}$};
            \node (C) at (2.75, -2) [label={[xshift=2pt]left:{\footnotesize $2^{k}$}}] {\footnotesize $x_1$};
            \node (D) at (-2.75, -4) [label=right:{\footnotesize $2^{k - 1}$}] {\footnotesize $x_{-2}$};
            \node (E) at (2.75, -4) [label=left:{\footnotesize $2^{k - 1}$}] {\footnotesize $x_2$};
            \node (F) at (-2.75, -9) [label={[xshift=-3pt]right:{\footnotesize $2$}}] {\footnotesize $x_{-k}$};
            \node (G) at (2.75, -9) [label={[xshift=3pt]left:{\footnotesize $2$}}] {\footnotesize $x_k$};
            \node (H) at (-2.75, -11) [label={[xshift=-3pt]right:{\footnotesize $1$}}] {\footnotesize $z_{-k}$};
            \node (I) at (2.75, -11) [label={[xshift=3pt]left:{\footnotesize $1$}}] {\footnotesize $z_k$};
            \node (J) at (-4.75, -4) [label={[xshift=3pt]left:{\footnotesize $1$}}] {\footnotesize $z_{-1}$};
            \node (K) at (4.75, -4) [label={[xshift=-3pt]right:{\footnotesize $1$}}] {\footnotesize $z_1$};
            \node (L) at (-4.75, -6) [label={[xshift=3pt]left:{\footnotesize $1$}}] {\footnotesize $z_{-2}$};
            \node (M) at (4.75, -6) [label={[xshift=-3pt]right:{\footnotesize $1$}}] {\footnotesize $z_2$};
            \node (N) at (-4.75, -9) [label={[xshift=3pt]left:{\footnotesize $1$}}] {\scriptsize $z_{1 - k}$};
            \node (O) at (4.75, -9) [label={[xshift=-3pt]right:{\footnotesize $1$}}] {\scriptsize $z_{k - 1}$};
            \node (P) at (-4.25, -13) [label={[xshift=3pt]left:{\footnotesize $\frac{1}{2}$}}] {\footnotesize $v_{-k}$};
            \node (Q) at (-1.25, -13) [label={[xshift=3pt]left:{\footnotesize $\frac{1}{2}$}}] {\footnotesize $y_{-k}$};
            \node (R) at (1.25, -13) [label={[xshift=-3pt]right:{\footnotesize $\frac{1}{2}$}}] {\footnotesize $v_k$};
            \node (S) at (4.25, -13) [label={[xshift=-3pt]right:{\footnotesize $\frac{1}{2}$}}] {\footnotesize $y_k$};

            \fill (-2.75, -6.25) circle (2pt);
            \fill (2.75, -6.25) circle (2pt);
            \fill (-2.75, -6.5) circle (2pt);
            \fill (2.75, -6.5) circle (2pt);
            \fill (-2.75, -6.75) circle (2pt);
            \fill (2.75, -6.75) circle (2pt);
            
            \draw[double] (A) -- (B);
            \draw[double] (A) -- (C);
            \draw[double] (B) -- (D);
            \draw[double] (C) -- (E);
            \draw[double] (D) -- (-2.75, -5.75);
            \draw[double] (E) -- (2.75, -5.75);
            \draw[double] (-2.75, -7.25) -- (F);
            \draw[double] (2.75, -7.25) -- (G);
            \draw[double] (F) -- (H);
            \draw[double] (G) -- (I);
            \draw (B) -- (J);
            \draw (C) -- (K);
            \draw (D) -- (L);
            \draw (E) -- (M);
            \draw (-3, -7.25) -- (N);
            \draw (3, -7.25) -- (O);
            \draw (H) -- (P);
            \draw (H) -- (Q);
            \draw (I) -- (R);
            \draw (I) -- (S);

            \end{tikzpicture}
        }
    \end{minipage}
    \hspace{0.125\textwidth}
    \begin{minipage}{0.2\textwidth}
        \centering
        \subfigure[]{
            \begin{tikzpicture}[scale=0.35]

            \tikzstyle{every node}=[circle, draw, minimum size=0.5cm, inner sep=0pt]

            \node (A) at (0, 0) [label=left:{\footnotesize $2^{k + 1}$}] {\footnotesize $y_0$};
            \node (B) at (-2.75, -2) [label={[xshift=-2pt]right:{\footnotesize $2^{k}$}}] {\footnotesize $y_{-1}$};
            \node (C) at (2.75, -2) [label={[xshift=2pt]left:{\footnotesize $2^{k}$}}] {\footnotesize $y_1$};
            \node (D) at (-2.75, -4) [label=right:{\footnotesize $2^{k - 1}$}] {\footnotesize $y_{-2}$};
            \node (E) at (2.75, -4) [label=left:{\footnotesize $2^{k - 1}$}] {\footnotesize $y_2$};
            \node (F) at (-2.75, -9) [label={[xshift=-3pt]right:{\footnotesize $2$}}] {\footnotesize $y_{-k}$};
            \node (G) at (2.75, -9) [label={[xshift=3pt]left:{\footnotesize $2$}}] {\footnotesize $y_k$};
            \node (H) at (-2.75, -11) [label={[xshift=-3pt]right:{\footnotesize $1$}}] {\footnotesize $z_{-k}$};
            \node (I) at (2.75, -11) [label={[xshift=3pt]left:{\footnotesize $1$}}] {\footnotesize $z_k$};
            \node (J) at (-4.75, -4) [label={[xshift=3pt]left:{\footnotesize $1$}}] {\footnotesize $z_{-1}$};
            \node (K) at (4.75, -4) [label={[xshift=-3pt]right:{\footnotesize $1$}}] {\footnotesize $z_1$};
            \node (L) at (-4.75, -6) [label={[xshift=3pt]left:{\footnotesize $1$}}] {\footnotesize $z_{-2}$};
            \node (M) at (4.75, -6) [label={[xshift=-3pt]right:{\footnotesize $1$}}] {\footnotesize $z_2$};
            \node (N) at (-4.75, -9) [label={[xshift=3pt]left:{\footnotesize $1$}}] {\scriptsize $z_{1 - k}$};
            \node (O) at (4.75, -9) [label={[xshift=-3pt]right:{\footnotesize $1$}}] {\scriptsize $z_{k - 1}$};
            \node (P) at (-4.25, -13) [label={[xshift=3pt]left:{\footnotesize $\frac{1}{2}$}}] {\footnotesize $x_{-k}$};
            \node (Q) at (-1.25, -13) [label={[xshift=3pt]left:{\footnotesize $\frac{1}{2}$}}] {\footnotesize $v_{-k}$};
            \node (R) at (1.25, -13) [label={[xshift=-3pt]right:{\footnotesize $\frac{1}{2}$}}] {\footnotesize $x_k$};
            \node (S) at (4.25, -13) [label={[xshift=-3pt]right:{\footnotesize $\frac{1}{2}$}}] {\footnotesize $v_k$};

            \fill (-2.75, -6.25) circle (2pt);
            \fill (2.75, -6.25) circle (2pt);
            \fill (-2.75, -6.5) circle (2pt);
            \fill (2.75, -6.5) circle (2pt);
            \fill (-2.75, -6.75) circle (2pt);
            \fill (2.75, -6.75) circle (2pt);
            
            \draw[double] (A) -- (B);
            \draw[double] (A) -- (C);
            \draw[double] (B) -- (D);
            \draw[double] (C) -- (E);
            \draw[double] (D) -- (-2.75, -5.75);
            \draw[double] (E) -- (2.75, -5.75);
            \draw[double] (-2.75, -7.25) -- (F);
            \draw[double] (2.75, -7.25) -- (G);
            \draw[double] (F) -- (H);
            \draw[double] (G) -- (I);
            \draw (B) -- (J);
            \draw (C) -- (K);
            \draw (D) -- (L);
            \draw (E) -- (M);
            \draw (-3, -7.25) -- (N);
            \draw (3, -7.25) -- (O);
            \draw (H) -- (P);
            \draw (H) -- (Q);
            \draw (I) -- (R);
            \draw (I) -- (S);
            
            \end{tikzpicture}
        }
    \end{minipage}
    \caption{Set $\Tau = \{T_1, T_2, T_3 \}$ of $z_0$-strategies for the Flower graph $J_m$, with $m = 2k + 1, m > 5$, where (a) shows $T_1$, (b) $T_2$, and (c) $T_3$. We omit the root $z_0$ in the representation of the subtrees, and we represent the trunk edges and branches by double lines and standard lines, respectively.}
    \label{fig:Jm_complete}
\end{figure}

\newpage

\section{Cube graph achieves the tight bound of Theorem~\ref{thm:min}}

The vertices of the cube graph $Q_d$ can be represented by binary strings of dimension $d$. Since $Q_d$ is vertex-transitive, we can choose, without loss of generality, the target $r$ as the binary string with all zero bits, denoted by $\bm{0}$. The eccentricity of the vertex $\bm{0}$ is $d$ and the binary string with all one bits, denoted as $\bm{1}$, is the only $\bm{0}$-peripheral vertex. The $j^{\text{th}}$ neighborhood of $\bm{0}$ consists of the $\binom{d}{j}$ vertices with exactly $j$ one bits. Note that every vertex $v \in V(Q_q - \{\bm{0}, \bm{1} \})$ is part of a shortest path from $\bm{0}$ to $\bm{1}$, satisfying the necessary condition to the tight bound of Theorem~\ref{thm:min}. Eq.~\eqref{eq:low} gives
\begin{equation}
    \label{eq:cube}
        \lambda_{Q_d, \bm{0}} \geq \sum_{j = 1}^{d} \max \left\{ \binom{d}{j}, 2^{d-j} \right\}
\end{equation}
since $I_{\text{sur}} = \left\{ 1 \leq j \leq e(r) - 1: 2^{d - j} > \binom{d}{j}\right\}$ and $2^d - 1 = \sum_{j = 1}^{d} \binom{d}{j}$. 

The certificate $\Tau$, provided by Hurlbert~\cite{WFL} and described as follows, shows the optimality of our bound. We have built a strategy $T \in \Tau$ as detailed below. We iterate the index $j$ from $1$ to $d$. For $j = 1$, we add the vertex $v_1 \in N(\bm{0})$ to be the only child of the root $\bm{0}$, setting $\omega_T(v_1) = 2^{d - 1}$. For $j > 1$, while $2^{d - j} > \binom{d}{j}$, we add a vertex $v_j \in N_j(\bm{0})$ to be the only child of $v_{ j - 1}$, setting $\omega_T(v_j) = 2^{d - j}$. Now, let $a_j = \left\lfloor \binom{d}{j}/2^{d - j}  \right\rfloor$ and $b_j = \binom{d}{j} \ \text{mod} \ 2^{d - j}$. For each $j$ such that $2^{d - j} \leq \binom{d}{j}$, we add $a_j + 1$ vertices of $N_j(r)$ in the subtree $T$, assigning weight $2^{d - j}$ for $a_j$ of them and weight $b_j$ to the remainder. All of these vertices must be connected to the vertices that have weight $2^{d - j + 1}$, which is ensured by the fact that each vertex of $N_{j - 1}(\bm{0})$ has $d - j + 1$ neighbors on $N_j(\bm{0})$, and $a_j < (d - j + 1) a_{j-1}$ whenever $d - j + 1 < d - 1$. This completes the construction of the strategy $T$. Each remaining strategy of $\Tau$ is generated by applying one of the $d!$ fixed permutations of the bit positions in the binary string of all vertices in the subtree $T$. Observe that $\sum_{v \in N_j(\bm{0})} \omega_\Tau(v) = 2^{d - j} \omega_{\text{min}}$ if $2^{d - j} > \binom{d}{j}$ and $\sum_{v \in N_j(\bm{0})} \omega_\Tau(v) = \binom{d}{j} \omega_{\text{min}}$ otherwise, where $\omega_{\text{min}} = d!$. Therefore, $\lambda_\Tau$ achieves the equality of Eq.~\eqref{eq:cube}, as claimed.

\section{Application of Theorem~\ref{thm:min} to the snark graphs}

Having the Eqs.~\eqref{eq:low} and~\eqref{eq:alt} in our arsenal, we compute the lower bound given by Theorem~\ref{thm:min} in the minimum WFL ratio $\lambda_{G, r}$ for Blanuša~2 and the Flower graphs. For $B_2$, since for any $r \in V(B_2)$, $|N(r)| = 3$, $|N_2(r)| = 6$, and $|N_3(r)| \geq 4$, we have $I_{\text{sur}} = \{1 \}$ and Eq.~\eqref{eq:low} gives $\lambda_{B_2, r} \geq 22$. 

On $J_m \ (m = 2k + 1, m \geq 3)$, we compute the size of each neighborhood through the distance between the target and the other vertices. For each target $x_0$, $v_0$, and $z_0$, we list these distances as follows.
\begin{itemize}
    \item Target $x_0$: for any $1 \leq i \leq k - 1$, $1 \leq j \leq k$, $d_{Jm}(x_0, z_0) = 1$, $d_{Jm}(x_0, v_0) = 2$, $d_{Jm}(x_0, y_0) = 2$, $d_{Jm}(x_0, x_j) = j$, $d_{Jm}(x_0, x_{-j}) = j$, $d_{Jm}(x_0, z_j) = j + 1$, $d_{Jm}(x_0, z_{-j}) = j + 1$, $d_{Jm}(x_0, v_j) = j + 2$, $d_{Jm}(x_0, v_{-j}) = j + 2$, $d_{Jm}(x_0, y_i) = i + 2$, $d_{Jm}(x_0, y_{-i}) = i + 2$, $d_{Jm}(x_0, y_k) = k + 1$, and $d_{Jm}(x_0, y_{-k}) = k + 1$.
    \item Target $v_0$: for any $\alpha \in \{x, y\}$ and $1 \leq j \leq k$, $d_{Jm}(v_0, z_0) = 1$, $d_{Jm}(v_0, \alpha_0) = 2$, $d_{Jm}(v_0, v_j) = j$, $d_{Jm}(v_0, v_{-j}) = j$, $d_{Jm}(v_0, z_j) = j + 1$, $d_{Jm}(v_0, z_{-j}) = j + 1$, $d_{Jm}(v_0, \alpha_j) = j + 2$, and $d_{Jm}(v_0, \alpha_{-j}) = j + 2$.
    \item Target $z_0$: for any $\alpha \in \{v, x, y\}$ and $1 \leq j \leq k$, $d_{Jm}(z_0, \alpha_0) = 1$, $d_{Jm}(z_0, \alpha_j) = j + 1$, $d_{Jm}(z_0, \alpha_{-j}) = j + 1$, $d_{Jm}(z_0, z_j) = j + 2$, and $d_{Jm}(z_0, z_{-j}) = j + 2$.
\end{itemize}

That way, the sizes of the neighborhoods are given by the following.
\begin{itemize}
    \item Target $x_0$: $|N(x_0)| = 3$, $|N_2(x_0)| = 6$, $|N_j(x_0)| = 8 \ (3 \leq j \leq k + 1)$, and $|N_{k + 2}(x_0)| = 2$.
    \item Target $v_0$: for $J_3$, we have $|N(v_0)| = 3$, $|N_2(v_0)| = 4$, $|N_3(v_0)| = 4$; for $J_5$, $|N(v_0)| = 3$, $|N_2(v_0)| = 6$, $|N_3(v_0)| = 6$, $|N_4(v_0)| = 4$; and for $J_m \ (m \geq 7)$, $|N(v_0)| = 3$, $|N_2(v_0)| = 6$, $|N_j(v_0)| = 8 \ (3 \leq j \leq k)$, $|N_{k + 1}(v_0)| = 6$, and $|N_{k + 2}(v_0)| = 4$.
    \item Target $z_0$: $|N(z_0)| = 3$, $|N_2(z_0)| = 6$, $|N_j(z_0)| = 8 \ (3 \leq j \leq k + 1)$, and $|N_{k + 2}(z_0)| = 2$.
\end{itemize}
Note that for all targets, we have $I_{\text{sur}} = \{1 \}$ for both $J_3$ and $J_5$, resulting from Eq.~\eqref{eq:low} in $\lambda_{J_3, r} \geq 12$ and $\lambda_{J_5, r} \geq 24$. For $J_m \ (m \geq 7)$, in all targets follows that $I_{\text{no}} = \{k, k+1, k+2 \}$ and $|N_{k}(r)| + |N_{k+1}(r)| + |N_{k+2}(r)| = 18$. In this case, since $|I_{\text{sur}}| \gg |I_{\text{no}}|$ when $m$ is large, we use Eq.~\eqref{eq:alt}, resulting in $\lambda_{J_m, r} \geq 2^{k + 2} + 10$. 

In comparison with the state of the art, obtained by Adauto et al.~\cite{WFL_snarks}---as lower bound on the pebbling numbers---we get the same results for $B_2$ and improve the minimum WFL lower bound by $1$ unit in $J_3$ and $2$ units in $J_m$ for $m \geq 5$. Although the gain was modest, we were able to establish the exact value of the minimum WFL ratio for the target $z_0$ on the graph $J_3$.

\section{Omitted certificates}

In this section, we exhibit the certificates for all targets of the Blanuša graph $B_2$ and Flower graphs $J_m \ (m = 2k + 1, m \geq 3)$. For all strategies shown in this section, we omit the root, and the vertices of the trunk are highlighted in bold. Tables~\ref{table:B2_strategies} and~\ref{table:J3_strategies} show the certificates of graphs $B_2$ and $J_3$, respectively. In these tables, in addition to the set $\Tau$ of strategies and its WFL ratio, we show the neighborhood of $r$, the $r$-peripheral vertices, $|\omega_T|$ and $\omega_{\text{min}}$.

\begin{table}[ht]
    \centering
    \begin{adjustbox}{width=1.1\textwidth,center=\textwidth}
    \scriptsize
    \renewcommand{\arraystretch}{1.2}
    \begin{tabular}{|c|c|c|c|c|c|c|c|c|}
        \hline
        \textbf{r} & \textbf{T} & \textbf{Subtrees} & \textbf{Weights} & \textbf{$N(r)$}  & \textbf{$P(r)$} & \textbf{$|\omega_T|$} & \textbf{$\omega_{\text{min}}$} & \textbf{$\lambda_T$} \\
        \hline
        \multirow{3}{*}{\textbf{$x_1$}}  
        & $T_1$ & $(\bm{z_1}, \bm{z_5}, z_2, \bm{z_5'}, z_2', \bm{x_5'}, \bm{z_1'}, \bm{x_2'})$ & $(8, 4, 5/2, 2, 5/4, 1, 1, 1/2)$ & \multirow{3}{*}{$\{z_1, x_3, x_4 \}$} & \multirow{3}{*}{$\{x_1', x_2', x_5', z_1' \}$} & \multirow{3}{*}{$58.5$} & \multirow{3}{*}{$2$} & \multirow{3}{*}{$29.25$} \\
        & $T_2$ & $(\bm{x_3}, \bm{z_3}, x_5, \bm{x_3'}, \bm{x_1'}, \bm{x_5'}, \bm{x_2'}, \bm{z_1'}, z_2')$ & $(8, 4, 2, 2, 1, 1, 1/2, 1/2, 1/4)$ &  &  &  &  & \\
        & $T_3$ & $(\bm{x_4}, \bm{z_4}, x_2, \bm{x_4'}, \bm{x_1'}, \bm{x_2'}, \bm{z_1'}, z_2')$ & $(8, 4, 2, 2, 1, 1, 1/2, 1/2)$ &  &  &  &  &  \\
        \hline
        \multirow{3}{*}{\textbf{$x_2$}}  
        & $T_1$ & $(\bm{z_2}, \bm{z_2'}, z_1, \bm{x_2'}, \bm{z_1'}, \bm{x_1'}, \bm{x_5'}, \bm{x_3'}, x_4', z_5')$ & $(8, 4, 5/2, 2, 2, 1, 1, 1/2, 1/2, 1/2)$ & \multirow{3}{*}{$\{z_2, x_5, x_4 \}$}  & \multirow{3}{*}{$\{x_4', x_5', z_2' \}$} & \multirow{3}{*}{$66.5$} & \multirow{3}{*}{$2.5$} & \multirow{3}{*}{$26.6$} \\
        & $T_2$ & $(\bm{x_5}, \bm{x_3}, \bm{z_5}, \bm{z_3}, \bm{z_5'}, \bm{x_3'}, \bm{x_5'}, \bm{x_1'}, x_2')$ & $(8, 4, 4, 2, 2, 1, 1, 1/2, 1/2)$ &  &  &  &  &  \\
        & $T_3$ & $(\bm{x_4}, \bm{z_4}, x_1, \bm{x_4'}, \bm{z_3}, \bm{x_1'}, \bm{x_3'}, \bm{x_5'}, z_1')$ & $(8, 4, 5/2, 2, 2, 1, 1, 1/2, 1/2)$ &  &  &  &  &  \\
        \hline
        \multirow{3}{*}{\textbf{$x_3$}} 
        & $T_1$ & $(\bm{z_3}, \bm{z_4}, \bm{x_3'}, \bm{x_4'}, \bm{x_1'}, \bm{x_2'}, x_5', \bm{z_1'}, \bm{z_2'})$ & $(8, 4, 5/2, 2, 5/4, 1, 1, 1/2, 1/2)$ & \multirow{3}{*}{$\{z_3, x_1, x_5 \}$} & \multirow{3}{*}{$\{x_2', z_1', z_2' \}$} & \multirow{3}{*}{$58.5$} & \multirow{3}{*}{$2$} & \multirow{3}{*}{$29.25$} \\
        & $T_2$ & $(\bm{x_1}, \bm{z_1}, x_4, \bm{z_2}, \bm{z_2'}, \bm{x_2'}, \bm{z_1'}, x_1')$ & $(8, 4, 2, 2, 1, 1/2, 1/2, 1/4)$ &  &  &  &  &  \\
        & $T_3$ & $(\bm{x_5}, \bm{z_5}, x_2, \bm{z_5'}, \bm{x_5'}, \bm{z_1'}, x_1', \bm{x_2'}, \bm{z_2'})$ & $(8, 4, 2, 2, 1, 1, 1/2, 1/2, 1/2)$ &  &  &  &  &  \\
        \hline
        \multirow{3}{*}{\textbf{$z_1$}}  
        & $T_1$ & $(\bm{x_1}, \bm{x_3}, \bm{x_4}, \bm{z_3}, \bm{z_4}, \bm{x_3'}, \bm{x_4'}, \bm{x_1'}, x_2', x_5')$ & $(8, 4, 4, 2, 2, 1, 1, 1/2, 1/2, 1/2)$ & \multirow{3}{*}{$\{x_1, z_2, z_5 \}$} & \multirow{3}{*}{$\{x_1', x_3', x_4' \}$} & \multirow{3}{*}{$66.5$} & \multirow{3}{*}{$2.5$} & \multirow{3}{*}{$26.6$} \\
        & $T_2$ & $(\bm{z_2}, \bm{z_2'}, x_2, \bm{x_2'}, \bm{z_1'}, \bm{x_1'}, \bm{x_4'}, \bm{x_3'}, z_4)$ & $(8, 4, 5/2, 2, 2, 1, 1, 1/2, 1/2)$ &  &  &  &  &  \\
        & $T_3$ & $(\bm{z_5}, \bm{z_5'}, x_5, \bm{x_5'}, \bm{z_1'}, \bm{x_1'}, \bm{x_3'}, \bm{x_4'}, z_3)$ & $(8, 4, 5/2, 2, 2, 1, 1, 1/2, 1/2)$ &  &  &  &  &  \\
        \hline
        \multirow{3}{*}{\textbf{$z_2$}}  
        & $T_1$ & $(\bm{z_2'}, \bm{z_1'}, \bm{x_1'}, x_2', \bm{x_3'}, x_4', x_5', \bm{z_3})$ & $(8, 4, 2, 2, 1, 1, 1, 1/2)$ & \multirow{3}{*}{$\{z_2', x_2, z_1 \}$} & \multirow{3}{*}{$\{x_3', z_3' \}$} & \multirow{3}{*}{$58$} & \multirow{3}{*}{$2$} & \multirow{3}{*}{$29$} \\
        & $T_2$ & $(\bm{x_2}, \bm{x_4}, x_1, \bm{z_4}, x_3, x_4', \bm{z_3}, \bm{x_3'})$ & $(8, 4, 2, 2, 1, 1, 1, 1/2)$ &  &  &  &  &  \\
        & $T_3$ & $(\bm{z_1}, \bm{z_5}, \bm{x_5}, \bm{z_5'}, \bm{x_3}, \bm{x_5'}, \bm{x_3'}, \bm{z_3})$ & $(8, 4, 2, 2, 1, 1, 1/2, 1/2)$ &  &  &  &  &  \\
        \hline
        \multirow{3}{*}{\textbf{$z_3$}}  
        & $T_1$ & $(\bm{z_4}, \bm{x_4}, \bm{x_4'}, \bm{x_2}, \bm{x_2'}, \bm{z_2}, \bm{z_2'}, z_1, z_1')$ & $(8, 4, 4, 2, 2, 1, 1, 1/2, 1/2)$ & \multirow{3}{*}{$\{z_4, x_3, x_3' \}$} & \multirow{3}{*}{$\{z_2, z_2' \}$} & \multirow{3}{*}{$68$} & \multirow{3}{*}{$2.5$} & \multirow{3}{*}{$27.2$} \\
        & $T_2$ & $(\bm{x_3}, \bm{x_1}, x_5, \bm{z_1}, z_5, \bm{z_2}, x_2, \bm{z_2'}, z_5')$ & $(8, 4, 4, 2, 2, 1, 1/2, 1/2, 1/2)$ &  &  &  &  &  \\
        & $T_3$ & $(\bm{x_3'}, \bm{x_1'}, x_5', \bm{z_1'}, z_5', \bm{z_2'}, x_2', \bm{z_2}, z_5)$ & $(8, 4, 4, 2, 2, 1, 1/2, 1/2, 1/2)$ &  &  &  &  &  \\
        \hline
    \end{tabular}
    \end{adjustbox}
    \caption{Set $\Tau$ of $r$-strategies $T_1$, $T_2$, and $T_3$ for all targets $r$ of the Blanuša graph $B_2$.}
    \label{table:B2_strategies}
\end{table}

\begin{table}[ht]
    \centering
    \begin{adjustbox}{width=1.1\textwidth,center=\textwidth}
    \scriptsize
    \renewcommand{\arraystretch}{1.2}
    \begin{tabular}{|c|c|c|c|c|c|c|c|c|}
        \hline
        \textbf{r} & \textbf{T} & \textbf{Subtrees} & \textbf{Weights} & \textbf{$N(r)$} & \textbf{$P(r)$} & \textbf{$|\omega_T|$} & \textbf{$\omega_{\text{min}}$} & \textbf{$\lambda_T$} \\
        \hline
        \multirow{3}{*}{\textbf{$x_0$}} 
        & $T_1$ & 
        $(\bm{z_0}, \bm{v_0}, y_0, \bm{v_1}, \bm{v_{-1}}, z_1, z_{-1})$ & $(4, 2, 2, 1, 1, 1/2, 1/2)$ & \multirow{3}{*}{$\{z_0, x_1, x_{-1} \}$} & \multirow{3}{*}{$\{v_1, v_{-1} \}$} & \multirow{3}{*}{$32$} & \multirow{3}{*}{$2.5$} & \multirow{3}{*}{$12.8$} \\
        & $T_2$ & $(\bm{x_1}, y_{-1}, \bm{z_1}, \bm{v_1}, y_1 \bm{v_{-1}}, y_0, v_0)$ & $(4, 2, 2, 1, 1/2, 1/2, 1/4, 1/4)$ &  &  &  &  &  \\
        & $T_3$ & $(\bm{x_{-1}}, y_1, \bm{z_{-1}}, \bm{v_{-1}}, y_{-1} \bm{v_{1}}, y_0, v_0)$ & $(4, 2, 2, 1, 1/2, 1/2, 1/4, 1/4)$ &  &  &  &  &  \\
        \hline
        \multirow{3}{*}{\textbf{$v_0$}}  
        & $T_1$ & $(\bm{z_0}, \bm{x_0}, \bm{y_0}, \bm{x_1}, \bm{x_{-1}}, \bm{y_1}, \bm{y_{-1}}, z_1, z_{-1})$ & $(4, 2, 2, 1, 1, 1, 1, 1/2, 1/2)$ & \multirow{3}{*}{$\{z_0, v_1, v_{-1} \}$} & \multirow{3}{*}{$\{x_1, x_{-1}, y_1, y_{-1} \}$} & \multirow{3}{*}{$32$} & \multirow{3}{*}{$2.5$} & \multirow{3}{*}{$12.8$} \\
        & $T_2$ & $(\bm{v_1}, \bm{z_1}, \bm{x_1}, \bm{y_1}, x_0, \bm{x_{-1}}, \bm{y_{-1}})$ & $(4, 2, 1, 1, 1/2, 1/2, 1/2)$ &  &  &  &  &  \\
        & $T_3$ & $(\bm{v_{-1}}, \bm{z_{-1}}, \bm{x_{-1}}, \bm{y_{-1}}, \bm{x_1}, y_0, \bm{y_1})$ & $(4, 2, 1, 1, 1/2, 1/2, 1/2)$ &  &  &  &  &  \\
        \hline
        \multirow{3}{*}{\textbf{$z_0$}}  
        & $T_1$ & $(\bm{v_0}, \bm{v_1}, \bm{v_{-1}}, \bm{z_1}, \bm{z_{-1}}, x_1, x_{-1}, y_1, y_{-1})$ & $(4, 2, 2, 1, 1, 1/2, 1/2, 1/2, 1/2)$ & \multirow{3}{*}{$\{v_0, x_0, y_0 \}$} & \multirow{3}{*}{$\{z_1, z_{-1} \}$} & \multirow{3}{*}{$36$} & \multirow{3}{*}{$3$} & \multirow{3}{*}{$12$} \\
        & $T_2$ & $(\bm{x_0}, \bm{x_1}, \bm{x_{-1}}, \bm{z_1}, \bm{z_{-1}}, y_1, y_{-1}, v_1, v_{-1})$ & $(4, 2, 2, 1, 1, 1/2, 1/2, 1/2, 1/2)$ &  &  &  &  &  \\
        & $T_3$ & $(\bm{y_0}, \bm{y_1}, \bm{y_{-1}}, \bm{z_1}, \bm{z_{-1}}, x_1, x_{-1}, v_1, v_{-1})$ & $(4, 2, 2, 1, 1, 1/2, 1/2, 1/2, 1/2)$ &  &  &  &  &  \\
        \hline
    \end{tabular}
    \end{adjustbox}
    \caption{Set $\Tau$ of $r$-strategies $T_1$, $T_2$, and $T_3$ for all targets $r$ of the Flower graph $J_3$.}
    \label{table:J3_strategies}
\end{table}

The certificates of $J_m \ (m = 2k + 1, m \geq 5)$ are listed as follows. Starting from target $x_0$, we have $N(x_0) = \{z_0, x_1, x_{-1} \}$ and $P(x_0) = \{v_k, v_{-k} \}$. The $x_0$-strategies are given by $\Tau = \{T_1, T_2, T_3 \}$, where

\newpage
\begin{itemize}[leftmargin=0em]
    \item $T_1(\bm{z_0}, \bm{v_0}, \bm{v_1}, \bm{v_{-1}}, \bm{v_2}, \bm{v_{-2}}, \ldots, \bm{v_k}, \bm{v_{-k}}, y_0, z_k, z_{-k}) \\ = (2^{k+1}, 2^k, 2^{k - 1}, 2^{k - 1}, 2^{k - 2}, 2^{k - 2}, \ldots, 1, 1, 5/2, 1/2, 1/2)$,
    \item $T_2(\bm{x_{1}}, \bm{x_{2}}, \ldots, \bm{x_{k - 1}}, z_{1}, z_{2}, \ldots, z_{k - 2}, \bm{x_{k}}, z_{k - 1}, y_{1}, y_{2}, \ldots, y_{k - 2}, \bm{z_{k}}, y_{-k}, y_{k - 1}, \bm{v_{k}}, y_{1 - k}, \bm{v_{-k}}, v_{k - 1}, y_{k})  \\ = (2^{k + 1}, 2^{k}, \ldots 8, 5, 5, \ldots, 5, 4, 3, 5/2, 5/2, \ldots, 5/2, 2, 2, 3/2, 1, 1, 1/2, 1/2, 1/2)$,
    \item $T_3(\bm{x_{-1}}, \bm{x_{-2}}, \ldots, \bm{x_{1 - k}}, z_{-1}, z_{-2}, \ldots, z_{2 - k}, \bm{x_{-k}}, z_{1 - k}, y_{-1}, y_{-2}, \ldots, y_{2 - k}, \bm{z_{-k}}, y_{k}, y_{1 - k}, \bm{v_{-k}}, y_{k - 1}, \bm{v_{k}}, v_{1 - k}, y_{-k}) \\ = (2^{k + 1}, 2^{k}, \ldots 8, 5, 5, \ldots, 5, 4, 3, 5/2, 5/2, \ldots, 5/2, 2, 2, 3/2, 1, 1, 1/2, 1/2, 1/2)$.
\end{itemize}

The total weight $|\omega_\Tau|$ is computed as
\begin{equation}
    \label{eqA1}
    \begin{split}
    |\omega_\Tau| & = \omega_\Tau(z_0) + \sum_{j = 1}^{k - 2} (\omega_\Tau(z_j) + \omega_\Tau(z_{-j})) + (\omega_\Tau(z_{k - 1}) + \omega_\Tau(z_{1 - k})) + (\omega_\Tau(z_k) + \omega_\Tau(z_{-k})) \\ & + \omega_\Tau(v_0) + \sum_{j = 1}^{k - 2} (\omega_\Tau(v_j) + \omega_\Tau(v_{-j})) + \sum_{j = k - 1}^{k} (\omega_\Tau(v_j) + \omega_\Tau(v_{-j})) + \sum_{j = 1}^{k} (\omega_\Tau(x_j) + \omega_\Tau(x_{-j})) \\ & + \omega_\Tau(y_0) + \sum_{j = 1}^{k} (\omega_\Tau(y_j) + \omega_\Tau(y_{-j})) \\ & = 2^{k + 1} + \sum_{j = 1}^{k - 2} 10 + 6 + 5 + 2^{k} + \sum_{j = 1}^{k - 2} 2^{k + 1 - j} + \sum_{j = k - 1}^{k} 5 + \sum_{j = 1}^{k} 2^{k + 3 - j} + 5/2 + \sum_{j = 1}^{k} 5 \\ & =\frac{13}{4} \ 2^{k + 2} + 15k - \frac{25}{2},
    \end{split}
\end{equation}
where we use the formula $\sum_{j = 1}^\alpha 2^{\beta - j} = 2^\beta - 2^{\beta - \alpha}$. Since $\omega_{\text{min}} = 2.5$, $\lambda_\Tau = 2^{k + 2} 13/10 + 6k - 5$. 

Now, for the target $v_0$, follows that $N(v_0) = \{z_0, v_1, v_{-1} \}$ and $P(v_0) = \{x_k, x_{-k}, y_k, y_{-k} \}$. The $v_0$-strategies are given by $\Tau = \{T_1, T_2, T_3 \}$, where
\begin{itemize}
    \item $T_1(\bm{z_0}, \bm{x_0}, \bm{y_0}, \bm{x_1}, \bm{x_{-1}}, \bm{y_1}, \bm{y_{-1}}, \bm{x_2}, \bm{x_{-2}}, \bm{y_2}, \bm{y_{-2}}, \ldots, \bm{x_k}, \bm{x_{-k}}, \bm{y_k}, \bm{y_{-k}}, z_k, z_{-k}) \\ = (2^{k + 1}, 2^{k}, 2^{k}, 2^{k - 1}, 2^{k - 1}, 2^{k - 1}, 2^{k - 1}, 2^{k - 2}, 2^{k - 2}, 2^{k - 2}, 2^{k - 2}, \ldots, 1, 1, 1, 1, 1/2, 1/2)$,
    \item $T_2(\bm{v_1}, \bm{v_2}, \ldots, \bm{v_k}, z_1, z_2, \ldots, z_{k - 1}, \bm{z_k}, \bm{x_k}, \bm{y_k}, \bm{x_{-k}}, \bm{y_{-k}}, x_{k - 1}, y_{k - 1}) \\ = (2^{k + 1}, 2^{k}, \ldots, 4, 5/2, 5/2, \ldots, 5/2, 2, 1, 1, 1/2, 1/2, 1/2, 1/2)$,
    \item $T_3(\bm{v_{-1}}, \bm{v_{-2}}, \ldots, \bm{v_{-k}}, z_{-1}, z_{-2}, \ldots, z_{1 - k}, \bm{z_{-k}}, \bm{x_{-k}}, \bm{y_{-k}}, \bm{x_{k}}, \bm{y_{k}}, x_{1 - k}, y_{1 - k}) \\ = (2^{k + 1}, 2^{k}, \ldots, 4, 5/2, 5/2, \ldots, 5/2, 2, 1, 1, 1/2, 1/2, 1/2, 1/2)$.
\end{itemize}
Let $A = \{x, y\}$. The total weight $|\omega_\Tau|$ is given by
\begin{equation}
    \label{eqA2}
    \begin{split}
    |\omega_\Tau| & = \sum_{\alpha \in A} \omega_\Tau(\alpha_0) + \sum_{\alpha \in A} \sum_{j = 1}^{k - 2} (\omega_\Tau(\alpha_j) + \omega_\Tau(\alpha_{-j})) + \sum_{\alpha \in A} \sum_{j = k - 1}^k (\omega_\Tau(\alpha_j) + \omega_\Tau(\alpha_{-j})) \\ & + \sum_{j = 1}^{k} (\omega_\Tau(v_j) + \omega_\Tau(v_{-j})) + \omega_\Tau(z_0) + \sum_{j = 1}^k (\omega_\Tau(z_k) + \omega_\Tau(z_{-k})) \\ & = \sum_{\alpha \in A} 2^{k} + \sum_{\alpha \in A} \sum_{j = 1}^{k - 2} 2^{k + 1 - j} + \sum_{\alpha \in A} \sum_{j = k - 1}^k 5 + \sum_{j = 1}^{k} 2^{k + 3 - j} + 2^{k + 1} + \sum_{j = 1}^k 5 \\ & =  2^{k + 1} + \sum_{j = 1}^{k - 2} 2^{k + 2 - j} + \sum_{j = k - 1}^k 10 + \sum_{j = 1}^{k} 2^{k + 3 - j} + 2^{k + 1} + \sum_{j = 1}^k 5 \\ & = 4 \ 2^{k + 2} + 5k - 4.
    \end{split}
\end{equation}
Since $\omega_{\text{min}} = 2.5$, $\lambda_\Tau = 2^{k + 2} 8/5 + 2k - 8/5$. 

To finish, for the target $z_0$, we have $N(z_0) = \{v_0, x_0, y_0 \} $ and $P(z_0) = \{z_k, z_{-k} \}$. The $z_0$-strategies are given by $\Tau = \{T_1, T_2, T_3 \}$, where
\begin{itemize}
    \item $T_1(\bm{v_0}, \bm{v_1}, \bm{v_{-1}}, \bm{v_2}, \bm{v_{-2}}, \ldots, \bm{v_k}, \bm{v_{-k}}, z_1, z_{-1}, z_2, z_{-2}, \ldots, z_{k-1}, z_{1 - k}, \bm{z_k}, \bm{z_{-k}}, x_k, x_{-k}, y_k, y_{-k}) \\ = (2^{k + 1}, 2^{k}, 2^{k}, 2^{k - 1}, 2^{k - 1}, \ldots, 1, 1, 1, 1, 1, 1, \ldots, 1, 1, 1, 1, 1/2, 1/2, 1/2, 1/2)$, 
    \item $T_2(\bm{x_0}, \bm{x_1}, \bm{x_{-1}}, \bm{x_2}, \bm{x_{-2}}, \ldots, \bm{x_k}, \bm{x_{-k}}, z_1, z_{-1}, z_2, z_{-2}, \ldots, z_{k-1}, z_{1 - k}, \bm{z_k}, \bm{z_{-k}}, v_k, v_{-k}, y_k, y_{-k}) \\ =(2^{k + 1}, 2^{k}, 2^{k}, 2^{k - 1}, 2^{k - 1}, \ldots, 1, 1, 1, 1, 1, 1, \ldots, 1, 1, 1, 1, 1/2, 1/2, 1/2, 1/2)$,
    \item $T_3(\bm{y_0}, \bm{y_1}, \bm{y_{-1}}, \bm{y_2}, \bm{y_{-2}}, \ldots, \bm{y_k}, \bm{y_{-k}}, z_1, z_{-1}, z_2, z_{-2}, \ldots, z_{k-1}, z_{1 - k}, \bm{z_k}, \bm{z_{-k}}, x_k, x_{-k}, v_k, v_{-k}) \\ = (2^{k + 1}, 2^{k}, 2^{k}, 2^{k - 1}, 2^{k - 1}, \ldots, 1, 1, 1, 1, 1, 1, \ldots, 1, 1, 1, 1, 1/2, 1/2, 1/2, 1/2)$.
\end{itemize}
Let $B = \{v, x, y\}$. The total weight $|\omega_\Tau|$ is computed as
\begin{equation}
    \label{eqA3}
    \begin{split}
    |\omega_\Tau| & =  \sum_{\beta \in B} \omega_\Tau(\beta_0) + \sum_{\beta \in B} \sum_{j = 1}^{k - 1} (\omega_\Tau(\beta_j) + \omega_\Tau(\beta_{-j})) + \sum_{\beta \in B} (\omega_\Tau(\beta_k) + \omega_\Tau(\beta_{-k})) + \sum_{j = 1}^k (\omega_\Tau(z_k) + \omega_\Tau(z_{-k})) \\ & = \sum_{\beta \in B} 2^{k + 1} + \sum_{\beta \in B} \sum_{j = 1}^{k - 1} 2^{k + 2 - j} + \sum_{\beta \in B} 6 + \sum_{j = 1}^k 6 \\ & = 3 \ 2^{k + 1} + 3 \sum_{j = 1}^{k - 1} 2^{k + 2 - j} + 18 + 6k \\ & = \frac{9}{2} \ 2^{k + 2} + 6k - 6,
    \end{split}
\end{equation} Since $\omega_{\text{min}} = 3$, $\lambda_\Tau = 2^{k + 2} 3/2 + 2k - 2$.

Note that the general certificate for the target $z_0$ on Flower graphs also applies for $J_3$, with this particular case shown in Table~\ref{table:J3_strategies}. This is not the case for the targets $x_0$ and $v_0$, whose certificates in Table~\ref{table:J3_strategies} differ from the general certificate for the class of Flower graphs.

\section{Blanuša $B_1$}

According to Adauto et al.~\cite{WFL_snarks}, analogous arguments can be done for the Blanuša snarks $B_1$ and $B_2$, so they present only arguments for $B_2$, claiming the upper bound $\pi(B_i) \leq 34$, for both $B_1$ and $B_2$. We have improved the upper bound for the pebbling number of the Blanuša snark $B_2$, establishing in Theorem~\ref{thm:Blan} that $\pi(B_2) \leq 30$. We establish next that $\pi(B_1) \leq 31$.

The Blanuša $B_1$ graph is shown in Figure~\ref{fig:B1}. There are five vertex classes on $B_1$, which can be represented, without loss of generality, as the targets $a_1$, $b_1$, $c_1$, $d_1$, and $e_1$. We apply our heuristic to $B_1$, obtaining the strategies shown in Table~\ref{table:B1_strategies}. In this table, in addition to the set $\Tau$ of strategies and its WFL ratio, we show the neighborhood of $r$, the $r$-peripheral vertices, $|\omega_T|$ and $\omega_{\text{min}}$. Note that $\lambda_{B_1} \leq 30$, which implies that $\pi(B_1) \leq 31$. Furthermore, with arguments similar to those used for $B_2$, we can conclude by Eq.~\eqref{eq:low} that Theorem~\ref{thm:min} gives the lower bound $\lambda_{B_1} \geq 22$.

\begin{figure}[ht]
    \centering
    \begin{tikzpicture}[scale=3]
    \tikzstyle{every node}=[circle, draw, minimum size=0.5cm, inner sep=0pt]

    \node (A0) at (0, 1) {\footnotesize $a_1$};
    \node (A1) at ({sin(360/17)}, {cos(360/17)}) {\footnotesize $b_1$};
    \node (A2) at ({sin(2*360/17)}, {cos(2*360/17)}) {\footnotesize $c_1$};
    \node (A3) at ({sin(3*360/17)}, {cos(3*360/17)}) {\footnotesize $b_2$};
    \node (A4) at ({sin(4*360/17)}, {cos(4*360/17)}) {\footnotesize $c_2$};
    \node (A5) at ({sin(5*360/17)}, {cos(5*360/17)}) {\footnotesize $d_1$};
    \node (A6) at ({sin(6*360/17)}, {cos(6*360/17)}) {\footnotesize $e_1$};
    \node (A7) at ({sin(7*360/17)}, {cos(7*360/17)}) {\footnotesize $d_2$};
    \node (A8) at ({sin(8*360/17)}, {cos(8*360/17)}) {\footnotesize $e_2$};
    \node (A9) at ({sin(9*360/17)}, {cos(9*360/17)}) {\footnotesize $e_2'$};
    \node (A10) at ({sin(10*360/17)}, {cos(10*360/17)}) {\footnotesize $d_2'$};
    \node (A11) at ({sin(11*360/17)}, {cos(11*360/17)}) {\footnotesize $e_1'$};
    \node (A12) at ({sin(12*360/17)}, {cos(12*360/17)}) {\footnotesize $d_1'$};
    \node (A13) at ({sin(13*360/17)}, {cos(13*360/17)}) {\footnotesize $c_2'$};
    \node (A14) at ({sin(14*360/17)}, {cos(14*360/17)}) {\footnotesize $b_2'$};
    \node (A15) at ({sin(15*360/17)}, {cos(15*360/17)}) {\footnotesize $c_1'$};
    \node (A16) at ({sin(16*360/17)}, {cos(16*360/17)}) {\footnotesize $b_1'$};
    \node (A17) at (0, 0) {\footnotesize $a_1'$};

    \draw (A0) -- (A1);
    \draw (A1) -- (A2);
    \draw (A2) -- (A3);
    \draw (A3) -- (A4);
    \draw (A4) -- (A5);
    \draw (A5) -- (A6);
    \draw (A6) -- (A7);
    \draw (A7) -- (A8);
    \draw (A8) -- (A9);
    \draw (A9) -- (A10);
    \draw (A10) -- (A11);
    \draw (A11) -- (A12);
    \draw (A12) -- (A13);
    \draw (A13) -- (A14);
    \draw (A14) -- (A15);
    \draw (A15) -- (A16);
    \draw (A16) -- (A0);
    \draw (A0) -- (A17);
    \draw (A17) -- (A3);
    \draw (A17) -- (A14);
    \draw (A1) -- (A13);
    \draw (A16) -- (A4);
    \draw (A2) -- (A7);
    \draw (A15) -- (A10);
    \draw (A5) -- (A9);
    \draw (A16) -- (A0);
    \draw (A12) -- (A8);
    \draw (A6) -- (A11);
    
\end{tikzpicture}
    \caption{Blanuša graph $B_1$.}
    \label{fig:B1}
\end{figure}
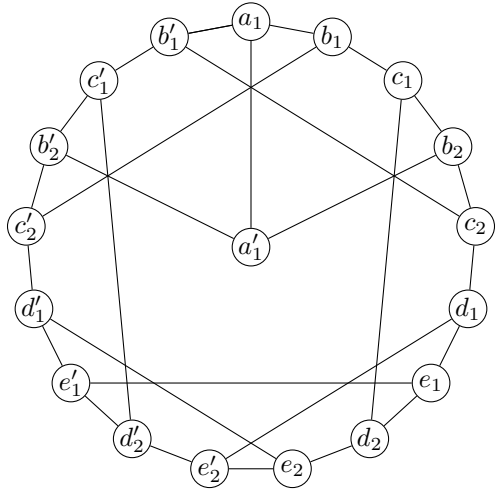

\begin{table}[ht]
    \centering
    \begin{adjustbox}{width=1.1\textwidth,center=\textwidth}
    \scriptsize
    \renewcommand{\arraystretch}{1.2}
    \begin{tabular}{|c|c|c|c|c|c|c|c|c|}
        \hline
        \textbf{r} & \textbf{T} & \textbf{Subtrees} & \textbf{Weights} & \textbf{$N(r)$}  & \textbf{$P(r)$} & \textbf{$|\omega_T|$} & \textbf{$\omega_{\text{min}}$} & \textbf{$\lambda_T$} \\
        \hline
        \multirow{3}{*}{\textbf{$a_1$}}  
        & $T_1$ & $(\bm{a_1'}, \bm{b_2}, \bm{b_2'}, \bm{c_2}, \bm{c_2'}, \bm{d_1}, \bm{d_1'}, \bm{e_1}, \bm{e_2}, \bm{e_1'}, \bm{e_2'})$ & $(8, 4, 4, 2, 2, 1, 1, 1/2, 1/2, 1/2, 1/2)$ & \multirow{3}{*}{$\{a_1', b_1, b_1' \}$} & \multirow{3}{*}{$\{e_1, e_2, e_1', e_2' \}$} & \multirow{3}{*}{$60$} & \multirow{3}{*}{$2$} & \multirow{3}{*}{$30$} \\
        & $T_2$ & $(\bm{b_1}, \bm{c_1}, \bm{d_2}, \bm{e_1}, \bm{e_2}, d_1, d_1', \bm{e_1'}, \bm{e_2'})$ & $(8, 4, 2, 1, 1, 1/2, 1/2, 1/2, 1/2)$ &  &  &  &  & \\
        & $T_3$ & $(\bm{b_1'}, \bm{c_1'}, \bm{d_2'}, \bm{e_1'}, \bm{e_2'}, d_1, d_1', \bm{e_1}, \bm{e_2})$ & $(8, 4, 2, 1, 1, 1/2, 1/2, 1/2, 1/2)$ &  &  &  &  &  \\
        \hline
        \multirow{3}{*}{\textbf{$b_1$}}  
        & $T_1$ & $(\bm{a_1}, \bm{b_1'}, a_1', \bm{c_2}, \bm{c_1'}, \bm{d_1}, \bm{d_2'}, \bm{e_2'})$ & $(8, 4, 5/2, 2, 2, 1, 1, 1/2)$ & \multirow{3}{*}{$\{a_1, c_1, c_2' \}$}  & \multirow{3}{*}{$\{d_1, d_2', e_2' \}$} & \multirow{3}{*}{$65$} & \multirow{3}{*}{$2.5$} & \multirow{3}{*}{$26$} \\
        & $T_2$ & $(\bm{c_1}, \bm{d_2}, b_2, \bm{e_1}, \bm{e_2}, \bm{d_1}, \bm{e_2'}, c_2, \bm{d_2'}, e_1')$ & $(8, 4, 5/2, 2, 2, 1, 1, 1/2, 1/2, 1/2)$ &  &  &  &  &  \\
        & $T_3$ & $(\bm{c_2'}, \bm{d_1'}, b_2', \bm{e_2}, \bm{e_1'}, \bm{d_2'}, \bm{e_2'}, c_1', \bm{d_1}, e_1)$ & $(8, 4, 5/2, 2, 2, 1, 1, 1/2, 1/2, 1/2)$ &  &  &  &  &  \\
        \hline
        \multirow{3}{*}{\textbf{$c_1$}} 
        & $T_1$ & $(\bm{b_1}, \bm{c_2'}, a_1, \bm{b_2'}, d_1', \bm{c_1'}, \bm{d_2'}) $ & $(8, 4, 2, 2, 2, 1, 1/2)$ & \multirow{3}{*}{$\{b_1, b_2, d_2 \}$} & \multirow{3}{*}{$\{c_1', d_2' \}$} & \multirow{3}{*}{$58$} & \multirow{3}{*}{$2$} & \multirow{3}{*}{$29$} \\
        & $T_2$ & $(\bm{b_2}, \bm{c_2}, a_1', \bm{b_1'}, \bm{d_1}, \bm{c_1'}, \bm{e_2'}, \bm{d_2'})$ & $(8, 4, 2, 2, 2, 1, 1, 1/2)$ &  &  &  &  &  \\
        & $T_3$ & $(\bm{d_2}, \bm{e_1}, e_2, \bm{e_1'}, \bm{d_2'}, e_2')$ & $(8, 4, 2, 2, 1, 1)$ &  &  &  &  &  \\
        \hline
        \multirow{3}{*}{\textbf{$d_1$}}  
        & $T_1$ & $(\bm{c_2}, \bm{b_2}, b_1', \bm{a_1'}, \bm{c_1}, a_1, \bm{b_1}, \bm{b_2'}, c_1', \bm{c_2'})$ & $(8, 4, 7/2, 2, 2, 7/4, 1, 1, 1/2, 1/2)$ & \multirow{3}{*}{$\{c_2, e_1, e_2' \}$} & \multirow{3}{*}{$\{b_1, b_2', c_2' \}$} & \multirow{3}{*}{$70.5$} & \multirow{3}{*}{$2.5$} & \multirow{3}{*}{$28.2$} \\
        & $T_2$ & $(\bm{e_1}, \bm{d_2}, \bm{e_1'}, \bm{c_1}, \bm{d_1'}, \bm{b_1}, \bm{c_2'}, a_1, \bm{b_2'})$ & $(8, 4, 4, 2, 2, 1, 1, 1/2, 1/2)$ &  &  &  &  &  \\
        & $T_3$ & $(\bm{e_2'}, \bm{d_2'}, \bm{e_2}, \bm{c_1'}, \bm{d_1'}, \bm{b_2'}, \bm{c_2'}, a_1', \bm{b_1}, a_1)$ & $(8, 4, 4, 2, 2, 1, 1, 1/2, 1/2, 1/4)$ &  &  &  &  &  \\
        \hline
        \multirow{3}{*}{\textbf{$e_1$}}  
        & $T_1$ & $(\bm{d_1}, \bm{c_2}, \bm{b_2}, \bm{b_1'}, \bm{a_1}, \bm{a_1'}, \bm{b_2'})$ & $(8, 4, 2, 2, 1, 1, 1/2)$ & \multirow{3}{*}{$\{d_1, d_2, e_1' \}$} & \multirow{3}{*}{$\{a_1, a_1', b_2' \}$} & \multirow{3}{*}{$58$} & \multirow{3}{*}{$2$} & \multirow{3}{*}{$29$} \\
        & $T_2$ & $(\bm{d_2}, \bm{c_1}, \bm{b_1}, e_2, \bm{a_1}, \bm{c_2'}, \bm{a_1'}, \bm{b_2'})$ & $(8, 4, 2, 2, 1, 1, 1/2, 1/2)$ &  &  &  &  &  \\
        & $T_3$ & $(\bm{e_1'}, \bm{d_2'}, \bm{c_1'}, d_1', e_2', \bm{b_2'}, c_2', \bm{a_1'})$ & $(8, 4, 2, 2, 2, 1, 1, 1/2)$ &  &  &  &  &  \\
        \hline
    \end{tabular}
    \end{adjustbox}
    \caption{Set $\Tau$ of $r$-strategies $T_1$, $T_2$, and $T_3$ for all targets $r$ of the Blanuša graph $B_1$. For each strategy, we omit the root and highlight the vertices of the trunk in bold.}
    \label{table:B1_strategies}
\end{table}

\end{document}